\documentclass[11pt]{article}

\usepackage{natbib}
\usepackage{fullpage}

\usepackage{environ}

\usepackage[utf8]{inputenc} %
\usepackage[T1]{fontenc}    %
\usepackage{hyperref}       %
\usepackage{url}            %
\usepackage{booktabs}       %
\usepackage{amsfonts}       %
\usepackage{nicefrac}       %
\usepackage{microtype}      %
\usepackage[ruled,vlined,noend]{algorithm2e}
\usepackage{amsmath,amssymb,amsfonts}
\usepackage{amsthm}
\usepackage{graphicx}
\usepackage{xcolor}

\newtheorem{proposition}{Proposition}

\theoremstyle{definition}
\newtheorem{example}{Example}
\theoremstyle{definition}
\newtheorem{definition}{Definition}

\renewcommand{\vec}{\mathbf}

\newcommand{\calA}{\mathcal{A}}
\renewcommand{\r}{\vec{r}}

\newcommand{\comp}{\mathsf{c}}

\newcommand{\nusers}{n}
\newcommand{\nitems}{m}
\newcommand{\latentdim}{d}

\newcommand{\R}{\mathbb{R}}
\newcommand{\calR}{\mathcal{R}}

\DeclareMathOperator*{\lse}{LSE}

\theoremstyle{definition}

\usepackage[labelfont=bf]{subcaption}
\usepackage[labelfont=bf]{caption}
\renewcommand{\cite}{\citep}

\usepackage{lmodern}

\graphicspath{{../icml2021/new_figures/}}

\title{Quantifying Availability and Discovery in Recommender Systems via Stochastic Reachability}

\author{
  Mihaela Curmei*, Sarah Dean*, and Benjamin Recht \\
  Department of EECS, University of California, Berkeley
}
\date{}

\begin{document}

\maketitle

\begin{abstract}
In this work, we consider how preference models in interactive recommendation systems determine the availability of content and users' opportunities for discovery. We propose an evaluation procedure based on stochastic reachability to quantify the maximum probability of recommending a target piece of content to an user for a set of allowable strategic modifications. This framework allows us to compute an upper bound on the likelihood of recommendation with minimal assumptions about user behavior. Stochastic reachability can be used to detect biases in the availability of content and diagnose limitations in the opportunities for discovery granted to users. We show that this metric can be computed efficiently as a convex program for a variety of  practical settings, and further argue that reachability is not inherently at odds with accuracy. We demonstrate evaluations of recommendation algorithms trained on large datasets of explicit and implicit ratings.
Our results illustrate how preference models, selection rules, and user interventions impact reachability and how these effects can be distributed unevenly.
\end{abstract}
\section{Introduction}

Through recommendation systems, personalized preference models mediate access to many types of information on the internet.
Aiming to surface content that will be consumed, enjoyed, and highly rated, these models are primarily designed to accurately predict individuals' preferences.
However, it is important to look beyond measures of accuracy towards notions of access.
The focus on improving recommender model accuracy favors systems in which human behavior becomes as predictable as possible---effects which have been implicated in unintended consequences like polarization or radicalization. %

We focus on questions of access and agency by
adopting
an \emph{interventional} lens, which considers arbitrary and strategic user actions.
We expand upon the notion of reachability first proposed by \citet{dean2020recommendations}, which measures the ability of an individual to influence a recommender model to select a certain piece of content.
We define a notion of stochastic reachability which quantifies
the maximum achievable likelihood of a given recommendation in the presence of strategic interventions.
This metric provides an upper bound on the ability of individuals to discover specific content, thus isolating unavoidable biases within preference models from those due to user behavior.

Our primary contribution is the definition of metrics based on stochastic reachability which capture the possible outcomes of a round of system interactions, including the \emph{availability} of content and \emph{discovery} possibilities for individuals.
In Section~\ref{sec:computing_reachability}, we show that they can be computed by solving a convex optimization problem for a class of relevant recommenders.
In Section~\ref{sec:geometry}, we draw connections between the stochastic and deterministic settings.
This perspective allows us to describe the relationship between agency and stochasticity and further to argue that there is not an inherent trade-off between reachability and model accuracy.
Finally, we present an audit of recommendation systems using a variety of datasets and preference models.
We explore how design decisions influence reachability and the extent to which biases in the training datasets are propagated.

\begin{figure}
    \center
    \includegraphics[width=0.7\columnwidth]{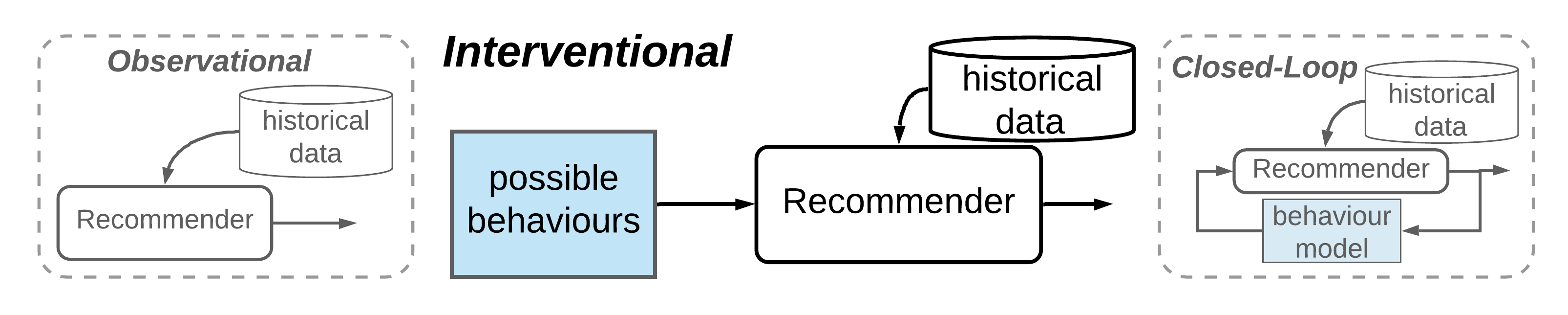}
    \caption{
    Conceptual framings of recommendation systems consider user behaviors to varying degrees. In this work we focus on evaluating interventional properties.} \label{fig:bias_diagram}
\end{figure}

\subsection{Related Work} %

The recommender systems literature has long proposed a variety of other metrics for evaluation, including notions of novelty, serendipity, diversity, and coverage~\cite{herlocker,castells2011novelty}.
There is a long history of measuring and mitigating bias in recommendation systems~\cite{chen2020bias}.
Empirical investigations have found evidence of popularity and demographic bias in domains including movies, music, books, and hotels~\cite{abdollahpouri2019impact,ekstrand2018all,ekstrand2018exploring,jannach2015recommenders}.
Alternative metrics are useful both for diagnosing biases and as objectives for post-hoc mitigating techniques such as calibration~\cite{steck2018calibrated} and re-ranking~\cite{singh2018fairness}.
A inherent limitation of these approaches is that they focus on \emph{observational} bias induced by preference models, i.e. examining the result of a single round of recommendations without considering individuals' behaviors.  While certainly useful, they fall short of providing further understanding into the interactive nature of recommendation systems.

The behavior of recommendation systems over time and in \emph{closed-loop} is still an open area of study. It is difficult to definitively link observational evidence of radicalization~\cite{ribeiro2020auditing,faddoul2020longitudinal} to proprietary recommendation algorithms.
Empirical studies of human behavior find mixed results on the relationship between recommendation and content diversity~\cite{nguyen2014exploring,flaxman2016filter}.
Simulation studies~\cite{chaney2018algorithmic,yao2021measuring,krauth2020offline} and theoretical investigations~\cite{dandekar2013polarization} shed light on phenomena  in simplified settings, showing how homogenization, popularity bias, performance, and polarization depend on assumed user behavior models.
Even ensuring accuracy
in sequential dynamic settings requires contending with closed-loop behaviors.
Recommendation algorithms must mitigate biased sampling in order to learn underlying user preference models,
using causal inference based techniques~\cite{schnabel2016recommendations,yang2018unbiased}
or
by balancing exploitation and exploration~\cite{kawale2015efficient,mary2015bandits}.
Reinforcement Learning algorithms contend with these challenges while considering a longer time horizon~\cite{chen2019top,ie2019slateq}, implicitly using data to exploit user behavior.

Our work eschews behavior models in favor of an \emph{interventional} framing which considers a variety of possible user actions.
Giving users control over their recommendations has been found to have positive effects, while reducing agency has negative effects~\cite{harper2015putting,lukoff2021design}.
The formal perspective we take on agency and access in recommender systems was first introduced by~\citet{dean2020recommendations},
and is closely related to a body on work on recourse in consequential decision making~\cite{ustun2019actionable,karimi2020survey}.
We build on this work to consider stochastic recommendation policies.

\begin{figure*}
	\center
	\includegraphics[height=3.25cm]{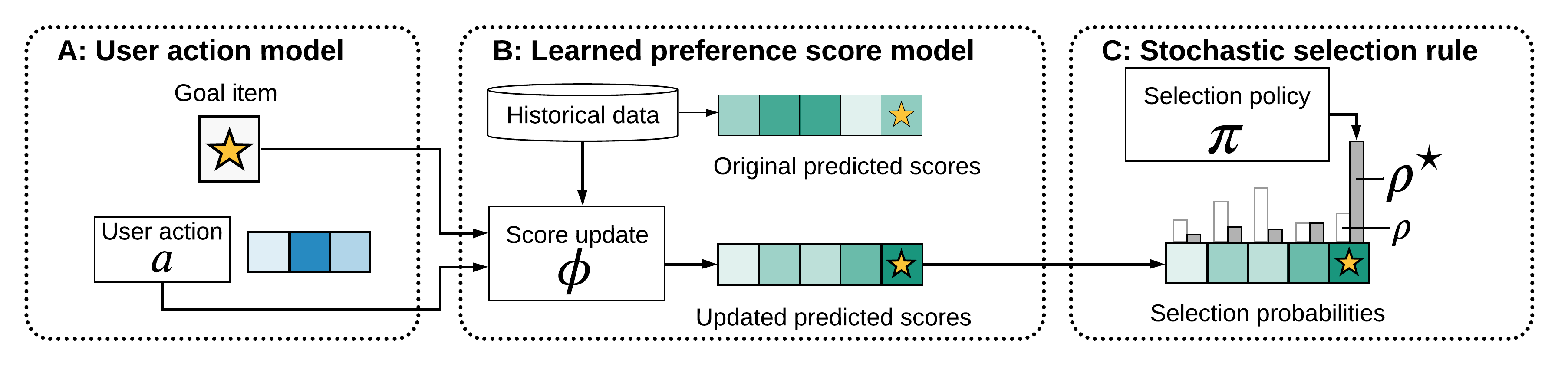}
	\caption{We audit recommender systems under a user action model (A), learned preference model (B), and stochastic selection rule (C). }
	\label{fig:pipeline}
\end{figure*}

\section{Metrics Based on Reachability} \label{sec:problem_setting}

\subsection{Stochastic Recommender Setting}

We consider systems composed of $\nusers$ individuals as well as a collection of $\nitems$ pieces of content.
For consistency with the recommender systems literature, we refer to individuals as users, pieces of content as items, and expressed preferences as ratings.
We will denote a rating by user $u$ of item $i$ as $r_{ui}\in \calR$, where $\calR\subseteq\R$ denotes the space of values which ratings can take. For example, ratings corresponding to the percentage of a video watched would have $\calR=[0,1]$ while discrete star ratings would have $\calR=\{1,2,3,4,5\}$.
The number of \emph{observed ratings} will generally be much smaller than the total number of possible ratings, and we denote by $\Omega_u\subseteq \{1,\dots,m\}$ the set of items seen by the user $u$.
The goal of a recommendation system is to understand the preferences of users and recommend relevant content.

In this work, we focus on the common setting in which recommenders are the composition of a  \emph{scoring function} $\phi$ with \emph{selection rule} $\pi$ (Figure~\ref{fig:pipeline}).
The scoring function models the preferences of users.
It is constructed based on historical data (e.g. observed ratings, user/item features) and returns a score for each user and item pair.
For a given user $u$ and item $i$, we denote $s_{ui}\in\R$
to be the associated score, and for user $u$ we will denote by $\vec s_u\in\R^\nitems$ the vector of scores for all items.
A common example of a scoring function
is a machine learning model which predicts future ratings based on historical data.

We will focus on the way that scores are updated after a round of user interaction.
For example, if a user consumes and rates several new items, the recommender system should update the scores in response.
Therefore, we parameterize the score function by an update rule, so that the new score vector is $\vec s_u^+ = \phi_u(\vec a)$, where $\vec a\in\calA_u$ represents actions taken by user $u$ and $\calA_u$ represents the set of all possible actions.
Thus $\phi_u$ encodes the historical data, the preference model class, and the update algorithm.
The action space $\calA_u$ represents possibilities for system interaction, encoding for example limitations due to user interface design.
We define the form of the score update function and discuss the action space in more detail in Section~\ref{sec:computing_reachability}.

The selection rule $\pi$ is a policy which, for given user $u$ and scores $\vec s_u$, selects one or more items from a set of specified \emph{target items} $\Omega_u^t\subseteq \{1,\dots,m\}$ as the next recommendation.
The simplest selection rule is a top-1 policy,
 which is a deterministic rule that selects the item with the highest score for each user. A simple stochastic rule is the $\epsilon$-greedy policy which with probability $1-\epsilon$ selects the top scoring item and with probability $\epsilon$ chooses uniformly from the remaining items.
Many additional approaches to recommendation can be viewed as the composition of a score function with a selection policy.
This setting also encompasses implicit feedback scenarios, where clicks or other behaviors are defined as or aggregated into ``ratings.''
Many recommendation algorithms, even those not specifically motivated by regression, include an intermediate score prediction step, e.g. point-wise approaches to ranking.
Further assumptions in Section~\ref{sec:computing_reachability} will not capture the full complexity of other techniques such as pairwise ranking and slate-based recommendations.
We leave such extensions to future work.

In this work, we are primarily interested in
stochastic policies which select items according to a probability distribution on the scores $\vec s_u$ parametrized by a exploration parameter.
Policies of this form are often used to balance exploration and exploration in online or sequential learning settings.
A stochastic selection rule recommends an item $i$ according to $ \mathbb{P}\left(\pi(\vec s_u, \Omega^t_u) = i\right)$, which is 0 for all  non-target items $i\notin\Omega_u^t$.
For example, to select among items that
have not yet been seen by the user, the target items are set as $\Omega_u^t=\Omega_u^{\comp}$ (recalling that $\Omega_u$ denotes the set of items seen by the user $u$).
Deterministic policies are
a special case of stochastic policies, with a degenerate distribution.

Stochastic policies have been proposed in the recommender system literature to improve diversity~\cite{christoffel2015blockbusters} or efficiently explore in a sequential setting~\cite{kawale2015efficient}.
By balancing exploitation of items with high predicted ratings against explorations of items with lower predictions, preferences can be estimated so that future predicted ratings are more accurate.
However, our work decidedly does not take a perspective based on accuracy.
Rather than supposing that users' reactions are predictable, we consider a perspective centered on agency and access.

\subsection{Reachability}
First defined in the context of recommendations by~\citet{dean2020recommendations}, an item $i$ is \emph{deterministically reachable} by a user $u$ if
there is some allowable modification to the user's ratings $\r_u$ that causes item to be recommended. Allowable modifications can include history edits, such as
removing or changing ratings of previously rated items. They can also include future looking modifications which assign ratings to a subset of unseen items.

In the setting where recommendations are made stochastically, we define an item $i$ to be \emph{$\rho$ reachable} by a user $u$ if
 there is some allowable action $\vec a$ such that {the updated probability that item $i$ is recommended after applying action $\vec a$ }; $\mathbb{P}(\pi(\phi_u(\vec a),\Omega_u^t)=i) $ is larger than $\rho$.
The maximum $\rho$ reachability for a user-item
 pair is defined as the solution to the following optimization problem:
 \begin{align}
	\begin{split}\label{eq:general_recourse}
	\rho^\star(u,i) = \max_{ \vec a \in \calA_u} ~~ &P(\pi(\phi_u(\vec a), \Omega_u^t)=i) .
	\end{split}
	\end{align}
We will also refer to $\rho^\star(u,i)$ as ``max reachability.''

For example, in the case of $\varepsilon$-greedy policy, $\rho^\star(u, i)=1-\varepsilon$ if item $i$ is deterministically reachable by user $u$, and is $\varepsilon / (|\Omega_u^{t}|-1)$ otherwise.

By measuring the maximum achievable probability of recommending an item to a user, we are characterizing a granular metric of \emph{access} within the recommender system.
It can also be viewed as an upper bound on the likelihood of recommendation with minimal assumptions about user behavior.
It may be illuminating to contrast this measure with a notion of expected reachability.
Computing expected reachability would require specifying the {distribution} over user actions, which would amount to modelling human behavior.
In contrast, max reachability requires specifying only the constraints arising from system design choices to define $\calA_u$ (e.g. the user interface).
By computing max reachability, we focus our analysis on the design of the recommender system, and avoid conclusions which are dependent on behavioral modelling choices.

Two related notions of user agency with respect to a target item $i$ are \emph{lift} and \emph{rank gain}. The lift measures the ratio between the maximum achievable probability of recommendation and the baseline:
\begin{align}
	\lambda^\star(u,i) = \frac{\rho^\star(u,i)}{\rho_0(u,i)}
\end{align}
where the baseline $\rho_0(u,i)$ is defined to capture the default probability of recommendation in the absence of strategic behavior, e.g. $P \left(\pi \left(\vec s_u, \Omega_u^t\right)=i\right)$.

The rank gain for an item $i$ is the difference in the ranked position of the item within the original list of scores $\vec s_u$ and its rank within the updated list of scores $\vec s_u^+$.

Lift and rank gain are related concepts, but ranked position is combinatorial in nature and thus difficult to optimize for directly.
They both measure agency because they compare the default behavior of a system to its behavior under a strategic intervention by the user.
Given that recommenders are designed with personalization in mind, we view the ability of users to influence the model in a positive light.
This is in contrast to much recent work in robust machine literature where strategic manipulation is undesirable.

\subsection{Diagnosing System Limitations}

The analysis of stochastic reachability can be used
to audit recommender systems and diagnose systemic biases {from an interventional perspective (Figure~\ref{fig:bias_diagram})}.
Unlike studies of observational bias, these analyses take into account system interactivity. Unlike studies of closed-loop bias, there is no dependence on a behavior model.
Because max reachability considers the best case over possible actions, it isolates structural biases from those caused in part by user behavior.

Max reachability is a metric defined for each user-item pair, and disparities across users and items can be detected through aggregations.
Aggregating over target items gives insight into a user's ability to discover content, thus detecting users who have been ``pigeonholed'' by the algorithm.
Aggregations over users can be used to compare how the system makes items available for recommendation.

We define the following user- and item-based aggregations:
\begin{align}
\begin{split}\label{eq:disc_avail}
D_u  = \sum_{i\in\Omega_u^t} \frac{\mathbf{1}\{ \rho_{ui}> \rho_t\} }{|\Omega_{u}^t|},~~
A_i = \frac{\sum_u \rho_{ui} \mathbf{1}\{ i\in \Omega_u^t\}}{\sum_u \mathbf{1}\{ i\in \Omega_u^t\}}
\end{split}
\end{align}

The discovery $D_u$ is the proportion of target items that have a high chance of being recommended, as determined by the threshold $\rho_t$. A natural threshold is the better-than-uniform threshold, $\rho_t=1/|\Omega_{u}^t|$, recalling that $\Omega_u^t$ is the set of target items. When $\rho_{ui}=\rho_0(u,i)$, baseline discovery counts the number of items that will be recommended with better-than-uniform probability and is determined by the spread of the recommendation distribution.
When $\rho_{ui}=\rho^\star(u,i)$, discovery counts the number of items that a user \emph{could} be recommended with better-than-uniform probability in the best case.
Low best-case discovery means that the recommender system inherently limits user access to content.

The item availability $A_i$ is the average likelihood of recommendation over all users who have item $i$ as a target.
It can be thought of as the chance that a uniformly selected user will be recommended item $i$. When $\rho_{ui}=\rho_0(u,i)$, the baseline availability measures the prevalence of the item in the recommendations. When $\rho_{ui}=\rho^\star(u,i)$, availability measures the prevalence of an item in the best case.
Low best-case availability means that the recommender system inherently limits the distribution of a given item.

\section{Computing Reachability}\label{sec:computing_reachability}

\subsection{Affine Recommendation}

In this section, we consider a restricted class of recommender systems for which the max reachability problem can be efficiently solved via convex optimization.

\paragraph{User action model}\label{sec:action_model}
We suppose that users interact with the system through expressed preferences, and thus
actions are updates to the vector $\r_u \in \calR^{\nitems}$,
a sparse vector of observed ratings.
For each user, the action model is based on distinguishing between \emph{action} and \emph{immutable} items.

Let $\Omega^{\mathcal A}_u$ denote the set of action items for which the ratings can be strategically modified by the user $u$. %
Then the action set $\calA_u = \calR^{|\Omega^{\mathcal A}_u|}$ corresponds to changing or setting the value of these ratings.
Figure \ref{fig:action_space} provides an illustration.
The action set should be defined to correspond to the interface
through which a user interacts with the recommender system.
For example, it could correspond to a display panel of ``previously viewed'' or ``up next'' items.

\begin{figure}
    \center
    \includegraphics[width=0.5\columnwidth]{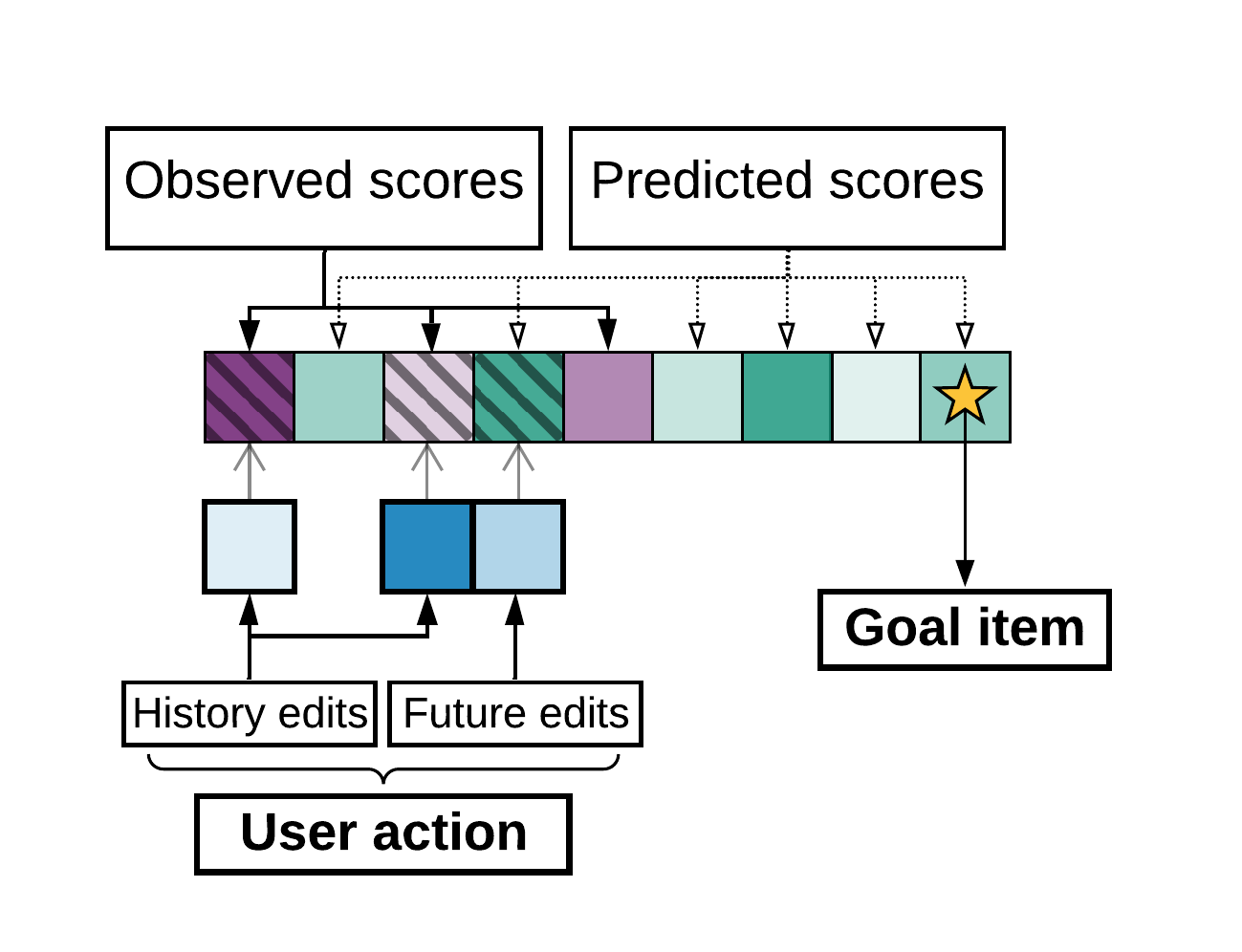}
    \caption{User action space: The shade represents the magnitude of historical (purple) or predicted (green) rating. The \emph{action items} are marked with diagonal lines; they can be strategically modified to maximize the recommendation probability of the \emph{goal item} (star). The value of the user action is shaded in blue.} \label{fig:action_space}
\end{figure}

The updated rating vector $ \r_u^+\in\calR^m$ is equal to $\r_u$ at the indices corresponding to immutable items and equal to the action $\vec a$ at the action items.
Note the partition into action and immutable is distinct from earlier partition of items into
observed and unobserved; action items can be both seen (history edits) and unseen (future reactions), as illustrated in
Figure \ref{fig:pipeline} (A).
For the reachability problem, we will consider a set of target items $\Omega_u^t$ that does not intersect with the action items $\Omega_u^\calA$. Depending on the specifics of the recommendation setting, we may also require that it does not intersect with the previously rated items $\Omega_u$.

We remark that additional user or item features used for scoring and thus recommendations could be incorporated into this
framework as either mutable or immutable features.
The only computational difficulty arises when mutable features are discrete or categorical.

\paragraph{Recommender model}
The recommender model is composed of a scoring function $\phi$ and a selection function $\pi$, which we now specify.
We consider \emph{affine score update functions} where for each user, scores are determined by an affine function of the action:
$\vec s_u^+=\phi_u(\vec a) = B_u\vec a + {\vec c}_{u}$
where $B_u\in\R^{\nitems\times |\Omega_u^\calA|}$ and $\vec c_u\in \R^{\nitems}$ are model parameters determined in part by historical data.
Such a scoring model arises from a variety of preference models, as shown in the examples in Section~\ref{sec:affine_examples}.

We now turn to the selection component of the recommender, which translates the score $\vec s_u$ into a probability distribution over target items.
The stochastic policy we consider is:
\begin{definition} Soft-max selection \\
	For $i\in\Omega_u^t$, the probability of item selection is given by
	\[P(\pi_\beta(\vec s_u,\Omega_u^t) = i) = \frac{e^{\beta s_{ui}}}{\sum_{j\in\Omega_u^t}e^{\beta s_{uj}}}\:.\]
\end{definition}

This form of stochastic policy samples an item according to a Boltzmann distribution defined by the predicted scores (Figure ~\ref{fig:pipeline}C). Distributions of this form are common in machine learning applications, and are known as Boltzmann sampling in reinforcement learning or online learning settings~\cite{wei2017reinforcement,cesa2017boltzmann}.

\subsection{Convex Optimization}

We now show that under affine score update models and soft-max selection rules, the maximum stochastic reachability problem can be solved by an equivalent convex problem.
First notice that for a soft-max selection rule with parameter $\beta$, we have that
\[\log\left(P(\pi_\beta(\vec s_u, \Omega_u^t)=i)\right) = \beta s_{ui} - \lse_{\substack{j\in\Omega_u^t}}\left(\beta s_{uj}\right)\]
where $\lse$ is the log-sum-exp function.

Maximizing stochastic reachability is equivalent to minimizing its negative log-likelihood.
Letting $\vec b_{ui}$ denote the $i$th row of the action matrix $B_u$ and substituting the form of the score update rule, we have the equivalent optimization problem:
\begin{align}
\begin{split}\label{eq:linear_recourse}
\min_{\vec a \in\calA_u} &  \lse_{\substack{j\in\Omega_u^t}} \left(\beta(\vec b_{uj}^\top \vec a+c_{uj})\right) - \beta (\vec b_{ui}^\top \vec a + c_{ui})
\end{split}
\end{align}
If the optimal value to~\eqref{eq:linear_recourse} is $\gamma^\star(u,i)$, then the optimal value for~\eqref{eq:general_recourse} is given by $\rho^\star(u,i) = e^{-\gamma^\star(u,i)}$.

The objective in~\eqref{eq:linear_recourse} is convex because log-sum-exp is a convex function, affine functions are convex, and the composition of a convex and an affine function is convex.
Therefore, whenever the action space $\calA_u$ is convex, so is the optimization problem.
The size of the decision variable scales with the dimension of the action, while the objective function relies on a matrix-vector product of size $|\Omega_u^t| \times |\calA_u|$.
Being able to solve the maximum reachability problem quickly is of interest, since auditing an entire system requires computing $\rho^\star$ for many user and item pairs.

\subsection{Examples} \label{sec:affine_examples}

In this section we review examples of common preference models and show how the score updates have an affine form.

\begin{example}
Matrix factorization models compute scores as rating predictions so that $S= PQ^\top$, where $P\in\mathbb{R}^{\nusers\times \latentdim}$ and $Q\in\mathbb{R}^{\nitems\times \latentdim}$ are respectively user and item factors for some latent dimension $\latentdim$. They are learned via the optimization
\[\min_{P,Q} \sum_{u}\sum_{i\in\Omega_u}\|\vec p_u^\top \vec q_i - r_{ui}\|_2^2 \:.\]
Under a stochastic gradient descent minimization scheme with step size $\alpha$, the one-step
update rule for a user factor is
\[\vec p^+_u = \vec p_u - \alpha \sum_{i\in\Omega_u^\calA}(\vec q_i\vec q_i^\top \vec p_u - \vec q_i r_{ui})\:,\]
Notice that this expression is affine in the action items.
Therefore, we have an affine score function:
\[\phi_u(\vec a) = Q\vec p_u^+=Q \left( \vec p_u - \alpha Q_\calA^\top Q_\calA \vec p_u - \alpha Q_\calA^\top \vec a \right)\]
where we define $Q_\calA = Q_{\Omega_u^\calA}\in \mathbb{R}^{|\Omega_u^\calA|\times d}$.
Therefore,
\[B_u = - \alpha Q Q_\calA^\top,\quad
\vec c_u = Q \left( \vec p_u - \alpha Q_\calA^\top Q_\calA \vec p_u\right)\:.\]
\end{example}

\begin{example}
Neighborhood models compute scores as rating predictions by a weighted average, with:
\[s_{ui} = \frac{\sum_{j\in\mathcal N_i} w_{ij} r_{uj}}{\sum_{j\in\mathcal N_i} |w_{ij}|}\]
where $w_{ij}$ are weights representing similarities between items and $\mathcal N_i$ is a set of indices of previously rated items in the neighborhood of item $i$. Regardless of the details of how these parameters are computed, the predicted scores are a linear function of observed scores: $\vec s_u = W \vec r_u$.

Therefore, the score updates take the form
\[\phi_u(\vec a) = W\vec r_u^+=\underbrace{W\vec r_u}_{\vec c_u} + \underbrace{W E_{\Omega_u^\calA}}_{B_u} \vec a\]
where $E_{\Omega_u^\calA}$ selects rows of $W$ corresponding to action items.
\end{example}

In both examples, the action matrices can be decomposed into two terms.
The first is a term that depends only on the preference model (e.g. item factors $Q$ or weights $W$), while the second is dependent on the user action model (e.g. action item factors $Q_\calA$ or action selector $E_{\Omega_u^\calA}$).

For simplicity of presentation, the examples above leave out model bias terms, which are common in practice.
Incorporating these model biases changes only the definition of the affine term in the score update expression.
We include the full action model derivation with biases in Appendix~\ref{app:examples}, along with additional examples. %

\section{Geometry of Reachability} \label{sec:geometry}

In this section, we explore the connection between stochastic and deterministic reachability to illustrate how both randomness and agency contribute to discovery as defined by the max reachability metric.
We then argue by example that it is possible to design preference models that guarantee deterministic reachability, and that doing so does not induce accuracy trade-offs.

\subsection{Connection to Deterministic Recommendation}
We now explore how the softmax style selection rule is a relaxation of top-$1$ recommendation.
For larger values of $\beta$, the selection rule distribution becomes closer to the deterministic top-1 rule.
This also means that the stochastic reachability problem can be viewed as a relaxation of the top-1 reachability problem.
In stochastic settings it is relevant to inquire the extent to which randomness impacts discovery and availability.
In the deterministic setting, the reachability of an item to a user is closely tied to agency---the ability of a user to influence their outcomes.
The addition of randomness induces exploration, but not in a way that is controllable by users.
In the following result, we show how this trade-off manifests in the max reachability metric itself.

\begin{proposition}\label{prop:stoch_det_connection}
Consider the stochastic reachability problem for a $\beta$-softmax selection rule as $\beta \to \infty$.
Then if an item $i$ is top-1 reachable by user $u$, $\rho^\star(u,i)\to1$. In the opposite case
that item $i$ \emph{is not} top-1 reachable, we have that $\rho^\star(u,i) \to0$.
\end{proposition}

\begin{proof}
Define
\[\gamma_\beta(\vec a) = \lse_{\substack{j\in\Omega_u^t}} \left(\beta\phi_{uj}(\vec a)  \right)  - \beta \phi_{ui}(\vec a) \]
and see that $\rho_{ui}(\vec a) = e^{-\gamma_\beta(\vec a)}$. 
Then we see that
\[\lim_{\beta\to\infty} \frac{1}{\beta}\gamma_\beta(\vec a) = \max_{\substack{j\notin\Omega_u}} \left(\phi_{uj}(\vec a)\right) - \phi_{ui}(\vec a) \] 
yields a top-$1$ expression.
If an item $i$ is top-1 reachable for user $u$, then there is some $\vec a$ such that the above expression is equal to zero. Therefore, as $\beta \to \infty$, $\gamma^\star \to 0 $, hence $\rho^\star\to1$. In the opposite case
	when an item \emph{is not} top-1 reachable we have that $\gamma^\star \to \infty$, hence $\rho^\star\to0$.
\end{proof}

This connection yields insight into the relationship between max reachability, randomness, and agency in stochastic recommender systems.
For items which are top-$1$ reachable, larger values of $\beta$ result in larger $\rho^\star$, and in fact the largest possible max reachability is attained as $\beta\to\infty$, i.e. there is no randomness.
On the other hand, if $\beta$ is too large, then items which are not top-$1$ reachable will have small $\rho^\star$.
There is some optimal finite $\beta\geq 0$ that maximizes $\rho^\star$ for top-1 unreachable items.
Therefore, we see a delicate balance when it comes to ensuring access with randomness.

Viewed in another light, this result says that for a fixed $\beta\gg1$, deterministic top-$1$ reachability ensures that $\rho^\star$ will be close to $1$.
We explore this perspective in the next section.

\subsection{Reachability Without Sacrificing Accuracy}

Specializing to affine score update models, we now highlight how parameters of the preference and action models play a role in determining max reachability.
Building on the connection to deterministic reachability, we make use of results about model and action space geometry from~\citet{dean2020recommendations}.
We recall the definition of the convex hull.

\begin{definition}[Convex hull]
The \emph{convex hull} of a set of vectors $\mathcal V = \{\vec v_i\}_{i=1}^n$ is defined as 
\[\mathrm{conv}\left(\mathcal V\right) = \left\{\sum_{i=1}^n w_i \vec v_i \mid \vec w\in\mathbb{R}^n_+,~~\sum_{i=1}^n w_i = 1\right\}\:.\]
A point $\vec v_j\in\mathcal V$ is a \emph{vertex} of the convex hull if
\[\vec v_j \notin \mathrm{conv}\left(\mathcal V\setminus \{\vec v_j\}\right)\:.\]
\end{definition}

\begin{proposition}\label{prop:conv_condition}
If $\vec b_{ui}$ is a vertex on the convex hull of $\{\vec b_{uj}\}_{j\in\Omega_u^t}$ and actions are real-valued,
then $\rho_{ui}^\star\to1$ as $\beta\to\infty$.
\end{proposition}

\begin{proof}
We begin by showing that if $\vec b_{ui}$ is a vertex on the convex hull of $\mathcal B = \{\vec b_{uj}\}_{j\in\Omega_u^t}$, then item $i$ is top-1 reachable.
This argument is similar to the proof of Results 1 and 2 in~\cite{dean2020recommendations}.

Item $i$ is top-1 reachable if there exists some $\vec a\in\R^{|\Omega_u^\calA|}$ such that $\vec b_{ui}^\top \vec a + c_{ui} \geq \vec b_{uj}^\top \vec a + c_{uj}$ for all $j\neq i$.
Therefore, top-1 reachability is equivalent to the feasibility of the following linear program
\begin{align*}
\min&~ 0^\top \vec a\\
\text{s.t.}&~  D_{ui}\vec a \geq \vec f_{ui}
\end{align*}
where $D_{ui}$ has rows given by $\vec b_{ui} - \vec b_{uj}$ and $\vec f_{ui}$ has entries given by $c_{uj}-c_{ui}$ for all $j\in\Omega_u^t$ with $j\neq i$.
Feasibility of this linear program is equivalent to boundedness of its dual:
\begin{align*}
\max&~  \vec f_{ui}^\top \vec \lambda\\
\text{s.t.}&~  D_{ui}^\top \vec \lambda =0,~~\vec \lambda\geq 0.
\end{align*}
We now show that if $\vec b_{ui}$ is a vertex on the convex hull of $\mathcal B$, then the dual is bounded because the only feasible solution is $\vec\lambda=0$.
To see why, notice that
\[D_{ui}^\top \vec \lambda =0 \iff \vec b_{ui} \sum_{\substack{j\in\Omega_u^t\\ j\neq i}} \lambda_j = \sum_{\substack{j\in\Omega_u^t\\ j\neq i}} \lambda_j\vec b_{uj} \]
If this expression is true for some $\vec\lambda\neq 0$, then we can write
\[\vec b_{ui}   = \sum_{\substack{j\in\Omega_u^t\\ j\neq i}} w_j\vec b_{uj},\quad w_j = \frac{\lambda_j}{\sum_{\substack{j\in\Omega_u^t\\ j\neq i}}\lambda_j } \implies \vec b_{ui} \in \mathrm{conv}\left(\mathcal B \setminus \{\vec b_{ui} \}\right)\:.\]
This is a contradiction, and therefore it must be that $\vec\lambda=0$ and therefore
the dual is bounded and item $i$ is top-1 reachable.

To finish the proof, we appeal to Proposition~\ref{prop:stoch_det_connection} to argue that since item $i$ is top-1 reachable,
then $\rho_{ui}^\star\to1$ as $\beta\to\infty$.
\end{proof}

This result highlights how the geometry of the score model determines when it is
preferable for the system to have minimal exploration, from the perspective of reachability.

We now consider whether relevant geometric properties of the model are predetermined by the goal of accurate prediction.
Is there a tension between ensuring reachability and accuracy?
We answer in the negative by presenting a construction for the case of matrix factorization models.
Our result shows that
the item and user factors ($P$ and $Q$) can be slightly altered such that all items become top-1 reachable at no loss of predictive accuracy.
The construction expands the latent dimension of the user and item factors by one and relies on a notion of sufficient richness for action items.

\begin{definition}[Rich actions]
For a set of item factors $\{\vec q_j\}_{j=1}^\nitems$, let $C=\max_j \|\vec q_j\|_2$.
Then a set of action items $\Omega_u^\calA\subseteq \{1,\dots,\nitems\}$ is \emph{sufficiently rich} if the vertical concatenation of their item factors and norms is full rank:
\[ \mathrm{rank}{\left(\begin{bmatrix}\vec q_i^\top & \sqrt{C^2 - \|\vec q_i\|_2^2} \end{bmatrix}_{i\in\Omega_u^\calA}\right)}
 = d+1 \:.\]
Notice that this can only be true if $|\Omega_u^\calA|\geq d+1$.
\end{definition}

\begin{proposition}\label{prop:mf_construction}
	Consider the MF model with user factors $P\in \mathbb{R}^{\nusers\times \latentdim}$ and item factors $Q\in \mathbb{R}^{\nitems\times \latentdim}$.
	Further consider any user $u$ with a sufficiently rich set of at least $\latentdim+1$ action items and real-valued actions.
	Then there exist $\tilde P\in\mathbb{R}^{n\times \latentdim+1}$ and $\tilde Q\in \mathbb{R}^{m\times \latentdim + 1}$ such that $PQ^\top = \tilde P \tilde Q^\top$ and under this model,
	$\rho^\star(u,i)\to1$ as $\beta\to\infty$ for all target items $i\in\Omega_u^t$.
\end{proposition}

\begin{proof}
	Let $C$ be the maximum row norm of $Q$ and define $\vec v\in\mathbb{R}^\nitems$ satisfying $v_i^2 = C^2 - \|\vec q_i\|_2^2$. 
	Then we construct modified item and user factors as
	\[\tilde Q =\begin{bmatrix} Q & \vec v\end{bmatrix}, \quad \tilde P =\begin{bmatrix} P & \vec 0\end{bmatrix}\:.\] 
	Therefore, we have that $\tilde P \tilde Q^\top = PQ^\top$.

	Then notice that by construction, each row of $\tilde Q$ has norm $C$, so each  $\tilde{\vec q}_i$ is on the boundary of the $\ell_2$ ball in $\mathbb{R}^{d+1}$. As a result, each $\tilde{\vec q}_i$ is a vertex on the convex hull of $\{\tilde{\vec q}_j\}_{j=1}^n$ as long as all $\vec q_j$ are unique. 

	For an arbitrary user $u$, the score model parameters are given by $\tilde{\vec b}_{ui} =\tilde Q_\calA\tilde{\vec q_i}$. 
	We show by contradiction that as long as the action items are sufficiently rich, each $\tilde{\vec b}_{ui}$ is a vertex on the convex hull of $\{\tilde{\vec b}_{uj}\}_{j=1}^n$.
	Supposing this is not the case for an arbitrary $i$,
	\[\tilde{\vec b}_{ui} = \sum_{\substack{j=1\\j\neq i}}^n w_j \tilde{\vec b}_{uj} \iff \tilde Q_\calA\tilde{\vec q_i} = \sum_{\substack{j=1\\j\neq i}}^n w_j \tilde Q_\calA\tilde{\vec q_i}
	\implies \tilde{\vec q_i} = \sum_{\substack{j=1\\j\neq i}}^n w_j \tilde{\vec q_i}\]
	where the final implication follows because the fact that $\tilde Q_\calA$ is full rank (due to richness) implies that $\tilde Q_\calA^\top \tilde Q_\calA$ is invertible.
	This is a contradiction, and therefore we have that each $\tilde{\vec b}_{ui}$ must be a vertex on the convex hull of $\{\tilde{\vec b}_{uj}\}_{j=1}^n$. 

	Finally, we appeal to Proposition~\ref{prop:conv_condition} to argue that $\rho^\star(u,i)\to1$ as $\beta\to\infty$ for all target items $i\in\Omega_u^t$.
	\end{proof}

The existence of such a construction demonstrates that there is not an unavoidable trade-off between accuracy and reachability in recommender systems.

\section{Audit Demonstration} \label{sec:experiments}
\subsection{Experimental Setup}
\paragraph{Datasets}
We evaluate\footnote{Reproduction code available at \url{github.com/modestyachts/stochastic-rec-reachability}} max  $\rho$ reachability in settings based on three popular recommendation datasets: MovieLens 1M (ML-1M)~\cite{harper2015movielens}, LastFM 360K~\cite{celma2010music} and MIcrosoft News Dataset (MIND)~\cite{wu2020mind}. ML-1M
is a dataset of 1 through 5 explicit ratings of movies, containing over one million recorded ratings; we do not perform any additional pre-processing.
LastFM is an implicit rating dataset containing the number of times a user has listened to songs of an artist. We used the version of the LastFM dataset
preprocessed by \citet{shakespeare2020exploring}.
For computational tractability, we select a random subset of $10\%$ of users and $10\%$ artists and define ratings as $r_{ui} = \log(\# \mathrm{listens}(u,i) + 1)$ to ensure that rating matrices are well conditioned. MIND is an implicit rating dataset containing clicks and impressions data. We use a subset of 50K users and 40K news articles spanning 17 categories and 247 subcategories. We transform news level click data into subcategory level aggregation and define the rating associated with a user-subcategory pair as a function of the number of times that the user clicked on news from that subcategory: $r_{ui}=\log(\# \mathrm{clicks}(u,i) + 1)$.
Table \ref{tbl:datasets} provides summary statistics and
Appendix \ref{app:data_details} contains further details about the datasets and preprocessing steps.

\begin{table}[t]
    \caption{Audit datasets}
    \label{tbl:datasets}
    \vskip 0.15in
    \begin{center}
    \begin{small}
    \begin{sc}
    \begin{tabular}{l|cccr}
    \toprule
    Data set & ML 1M & LastFM 360K & MIND \\
    \midrule
    Users     & 6040 & 13698 & 50000\\
    Items     & 3706& 20109& 247 \\
    Ratings    & 1000209 & 178388 & 670773\\
    Density (\%) & 4.47\% & 0.065\% & 5.54\%\\
    LibFM rmse    & 0.716& 1.122&    0.318     \\
    KNN rmse     &0.756 &  1.868& - \\
    \bottomrule
    \end{tabular}
    \end{sc}
    \end{small}
    \end{center}
    \vskip -0.1in
    \end{table}

\paragraph{Preference models}
We consider two preference models: one based on matrix factorization (MF) as well as a neighborhood based model (KNN). We use the LibFM SGD implementation \cite{rendle:tist2012} for the MF model and use the item-based k-nearest neighbors model implemented by~\citet{krauth2020offline}.  For each dataset and recommender model we perform hyper-parameter tuning using a 10\%-90\% test-train split. We report test performance in Table \ref{tbl:datasets}. See Appendix \ref{app:model_tuning} for details about tuning.
Prior to performing the audit, we retrain the recommender models with the full dataset.

\paragraph{Reachability experiments}
To compute reachability, it is further necessary to specify additional elements of the recommendation pipeline: the user action model, the set of target items, and the soft-max selection parameter.

We consider three types of user action spaces: \emph{History Edits}, \emph{Future Edits}, and \emph{Next K} in which users can strategically modify the ratings associated to $K$ randomly chosen items from their history, $K$ randomly chosen unobserved items, or the top-$K$ items according to the baseline scores of the preference model. For each of the action spaces we consider a range of $K$ values. We further constrain actions to lie in an interval corresponding to the rating range, using $[1,5]$ for movies and $[0,10]$ for music and news.

In the case of movies (ML-1M) we consider target items to be all items that are neither seen nor action items. In the case of music and news recommendations (LastFM \& MIND), the target items are all the items with the exception of action items. %
This reflects an assumption that music created by a given artist or news within a particular subcategory can be consumed repeatedly, while movies are viewed once.

For each dataset and recommendation pipeline, we compute max reachability for soft-max selection rules parametrized by a range of $\beta$ values.
Due to the computational burden of large dense matrices, we compute metrics for a subset of users and target items sampled uniformly at random. For details about runtime, see Appendix \ref{app:reach_complexity}.

\begin{figure}
    \center
    \includegraphics[width=0.8\columnwidth]{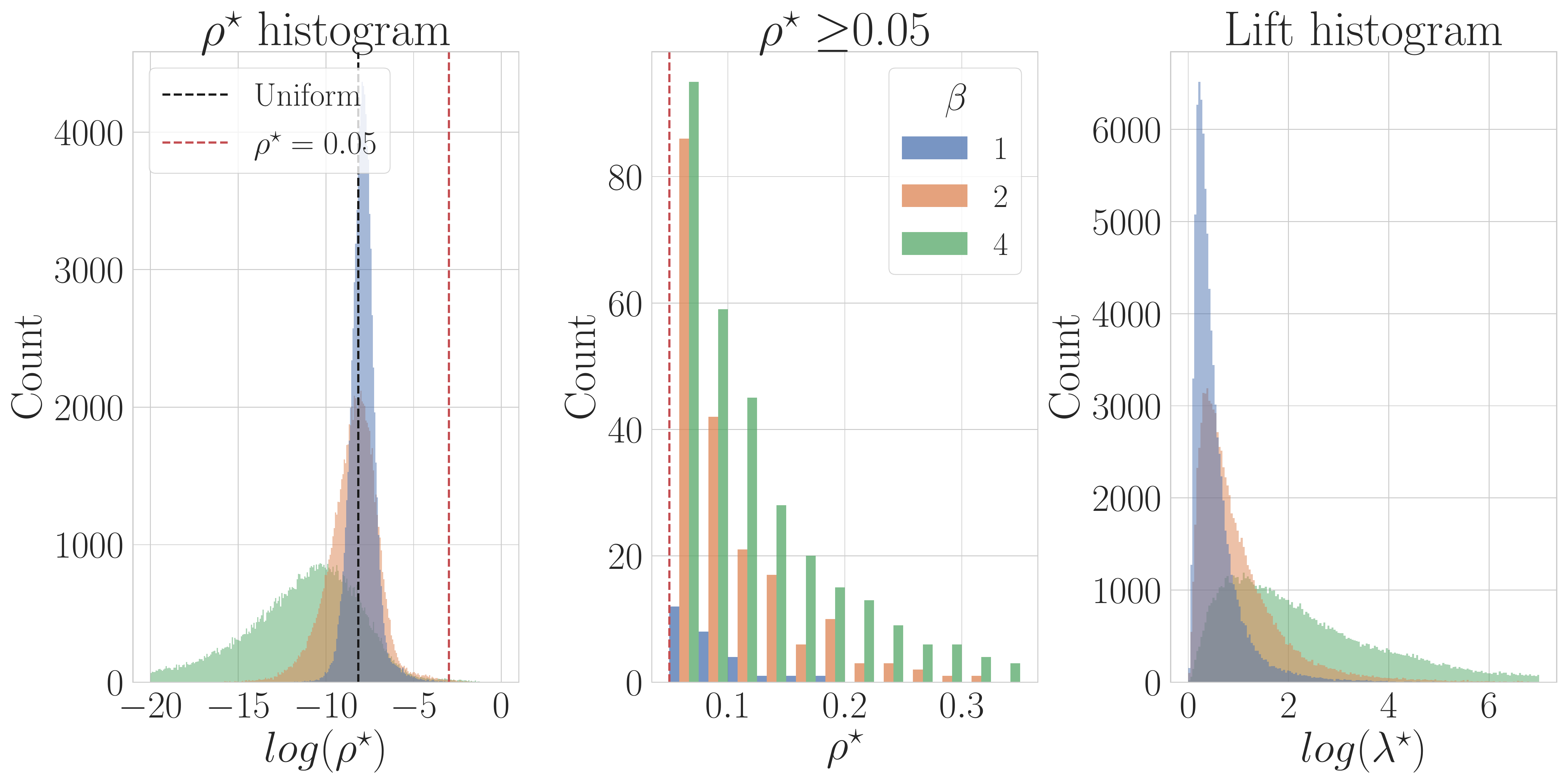}
    \caption{Left: Histogram of log max reachability values for $\beta = [1,2,4]$. Black dotted line denotes $\rho^\star$ for uniformly random recommender.  Center: Histogram of $\rho^\star>0.05$ (red dotted line). Right: Histogram of log-lifts.
    Reachability evaluated on ML-1M for $K=5$ Random Future action space and a LibFM model.} \label{fig:beta_hist}
\end{figure}

\begin{figure}
    \center
    \includegraphics[width=0.8\columnwidth]{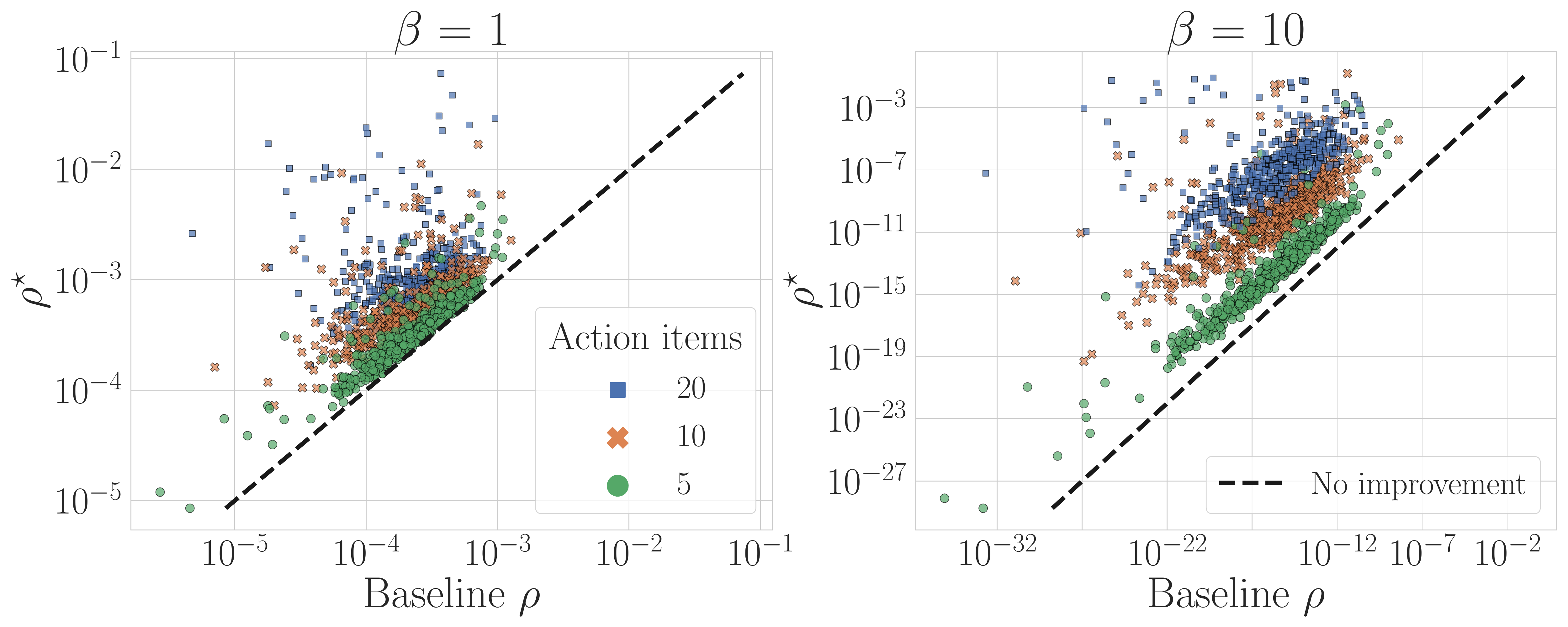}
    \caption{ Log scale scatterplot of $\rho^\star$ values against baseline $\rho$ for $K \in [5,10,20]$. Colors indicate action space size $K$. We compare low (left) and high (right) stochasticity. Reachability evaluated on ML-1M for Random Future action space and a LibFM model.} \label{fig:k_beta_scatter}
\end{figure}

\subsection{Impact of Recommender Pipeline} \label{sec:pipeline}
We begin by examining the role of recommender pipeline components: stochasticity of item selection, user action models, and choice of preference model. All presented experiments in this section use the ML-1M dataset.

These experiments show that more stochastic recommendations correspond to higher average max reachability values, whereas more deterministic recommenders have a more disparate impact, with a small number of items achieving higher $\rho^\star$.
We also see that the impact of the user action space differs depending on the preference model. For neighborhood based preference models, strategic manipulations to the history are most effective at maximizing reachability, whereas manipulations of the items most likely to be recommended next are ineffective.

\paragraph{Role of stochasticity}
We investigate the role of the $\beta$ parameter in the item selection policy. Figure \ref{fig:beta_hist} illustrates the relationship between the stochasticity of the selection policy and max reachability. There are significantly more target items with better than random reachability for low values of $\beta$. However, higher values of $\beta$ yield more items with high reachability potential ($>5\%$ likelihood of recommendation). These items are typically items that are top-1 or close to top-1 reachable. While lower $\beta$ values provide better reachability on average and higher $\beta$ values provide better reachability at the ``top'', higher $\beta$ uniformly out-performs lower $\beta$ values in terms of the lift metric. This suggests that larger $\beta$ corresponds to more user agency, since the relative effect of strategic behavior is larger. However, note that for very large values of $\beta$, high lift values are not so much the effect of improved reachability as they are due to very low baseline recommendation probabilities.

\paragraph{Role of user action model}
We now consider different action space sizes.
In Figure \ref{fig:k_beta_scatter} we plot max reachability for target items of a particular user over varying levels of selection rule stochasticity and varying action space sizes. Larger action spaces correspond to improved item reachability for all values of $\beta$. However, increases in the number of action items have a more pronounced effect for larger $\beta$ values.

While increasing the size of the action space uniformly improves reachability, the same cannot be said about the type of action space.
For each user, we compute the average lift over target items as a metric for user agency in a recommender (Figure \ref{fig:knn_libfm_lift}).
For LibFM, the choice of action space does not strongly impact the average user lift, though \emph{Next K} displays more variance across users than the other two. However, for Item KNN, there is a stark difference between \emph{Next K} and and random action spaces.

\begin{figure}
    \center
    \includegraphics[width=0.8\columnwidth]{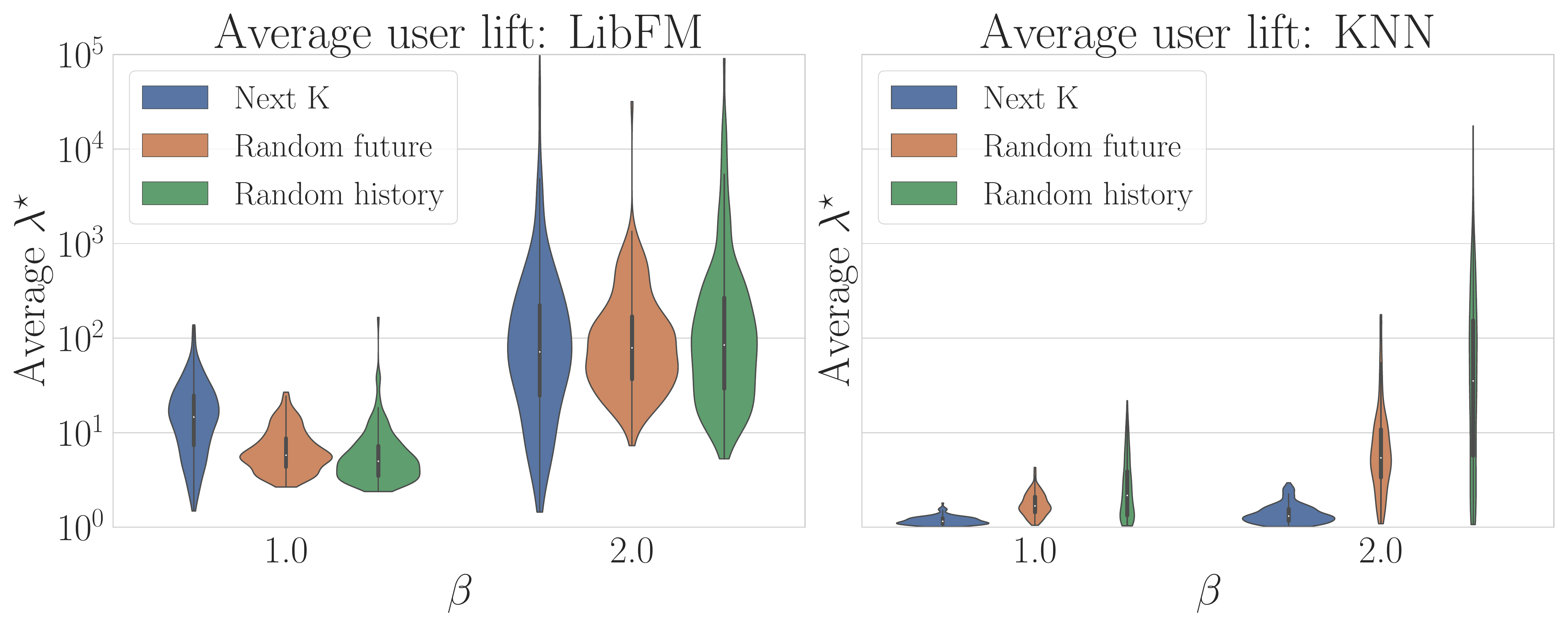}
    \caption{The distribution of average lifts (a notion of agency) over users. Colors indicate different user action spaces for LibFM (left) and KNN (right) on ML-1M.} \label{fig:knn_libfm_lift}
\end{figure}

\paragraph{Role of preference model}
As Figure \ref{fig:knn_libfm_lift} illustrates, a system using LibFM provides more agency on average than one using KNN. We now consider how this relates to properties of the preference models.
First, consider the fact that for LibFM, there is higher variance among user-level average lifts observed for \emph{Next K} action space compared with random action spaces.
This can be understood as resulting from the user-specific nature of \emph{Next K} recommended items. On the other hand, random action spaces are user independent, so it is not surprising that there is less variation across users.

In a neighborhood-based model users have leverage to increase the $\rho$ reachability only for target items in the neighborhood of action items. In the case of KNN, the next items up for recommendation are in close geometrical proximity to each other. This limits the opportunity for discovery of more distant items for \emph{Next K} action space. On the other hand, the action items are more uniformly over space of item ratings in random action models, thus contributing to much higher opportunities for discovery.
Additionally, we see that \emph{History Edits} displays higher average lift values than \emph{Future Edits}. We posit that this is due to the fact that editing $K$ items from the history leads to a larger ratio of strategic to non-strategic items. %

\subsection{Bias in Movie, Music, and News Recommendation}\label{sec:exp_bias}
{We futher compare aggregated stochastic reachability against properties of user and items to investigate bias.
We aggregate baseline and max reachability to compute user-level metrics of discovery and item-level metrics of availability.
The audit demonstrates popularity bias for items with respect to baseline availability. This bias persists in the best case for neighborhood based recommenders and is thus unavoidable, whereas it could be mitigated for MF recommenders. User discovery aggregation reveals inconclusive results with weak correlations between the length of users' experience and their ability to access content.}

\begin{figure}
    \center
    \includegraphics[width=0.8\columnwidth]{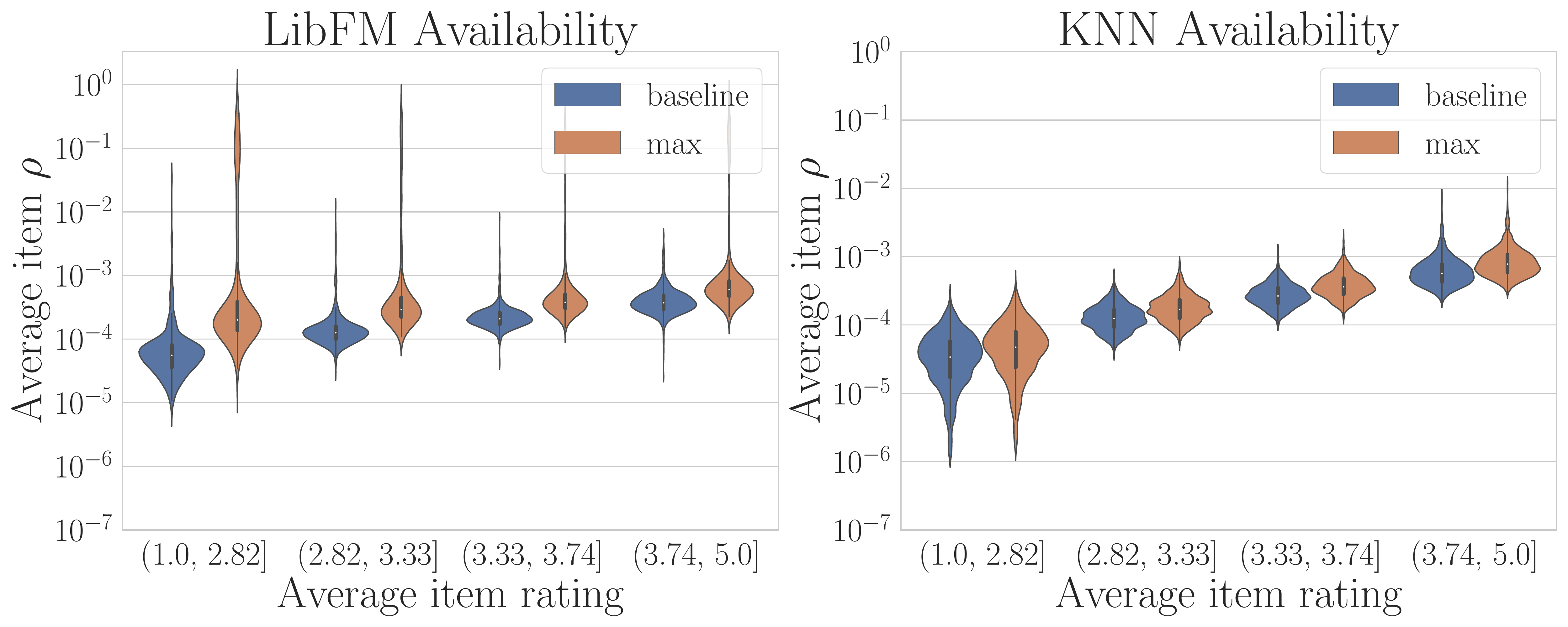}
    \caption{Comparison of baseline and best case availability of content, across four popularity categories for LibFM (left) and KNN (right) preference models. Reachability evaluated on ML-1M for \emph{Next 10} action space with $\beta=2$.} \label{fig:pop_bias}
\end{figure}

\paragraph{Popularity bias}
In Figure~\ref{fig:pop_bias}, we plot the baseline and best case item availability (as in~\eqref{eq:disc_avail}) to investigate popularity bias.
We consider popularity defined by the average rating of an item in a dataset.
Another possible definition of popularity is rating frequency, but for this definition we did not observe any discernable bias.
For both LibFM and KNN models, the baseline availability displays a correlation with item popularity, with Spearman's rank-order correlations of $r_s=0.87$ and $r_s=0.95$.
This suggests that as recommendations are made and consumed, more popular items will be recommended at disproportionate rates.

Furthermore, the best case availability for KNN displays a similar trend ($r_s=0.94$), indicating that the propagation of popularity bias can occur independent of user behavior.
This does not hold for LibFM, where the best case availability is less clearly correlated with popularity ($r_s=0.57$).
The lack of correlation for best case availability holds in the additional settings of music artist and news recommendation with the LibFM model (Figure~\ref{fig:pop_music_news}).
Our audit does not reveal an unavoidable systemic bias for LibFM recommender, meaning that any biases observed in deployment are due in part to user behaviour.
In contrast, we see a systematic bias for the KNN recommender, meaning that regardless of user actions, the popularity bias will propagate.

\begin{figure}
    \center
    \includegraphics[width=0.8\columnwidth]{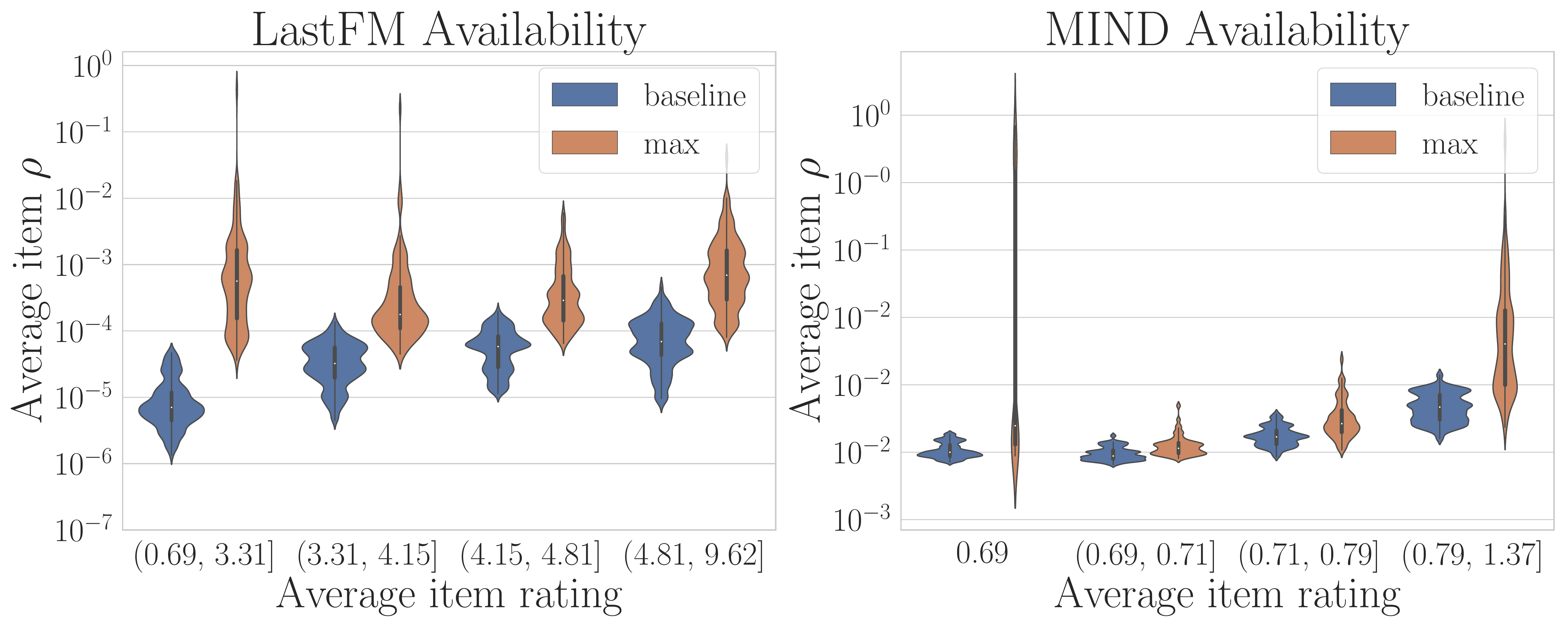}
    \caption{Comparison of baseline and best case availability of content for four popularity categories for LastFM (left) and MIND (right) with \emph{Next 10} actions, LibFM model, and $\beta=2$. } \label{fig:pop_music_news}
\end{figure}

\paragraph{Experience bias}
To consider the opportunities for discovery provided to users, we perform user level aggregations of max reachability values as in~\eqref{eq:disc_avail}.
We investigate experience bias by considering how the discovery metric changes as a function of the number of different items a user has consumed so far, i.e. their experience.
Figure~\ref{fig:exp_bias} illustrates that experience is weakly correlated with baseline discovery for movie recommendation ($r_s=0.48$), but not so much for news recommendation ($r_s=0.05$).
The best case discovery is much higher, meaning that users have the opportunity to discover many of their target items.
However, the weak correlation with experience remains for best case discovery of movies ($r_s=0.53$).

\section{Discussion} \label{sec:conclusion}

In this paper, we generalize reachability as first defined by~\cite{dean2020recommendations} to incorporate stochastic recommendation policies.
We show that for linear preference models and soft-max item selection rules, max reachability can be computed via a convex program for a range of user action models.
Due to this computational efficiency, reachability analysis can be used to audit recommendation algorithms.
Our experiments illustrate the impact of system design choices and historical data on the availability of content and users' opportunities for discovery, highlighting instances in which popularity bias is inevitable regardless of user behavior.

\begin{figure}
    \center
    \includegraphics[width=0.8\columnwidth]{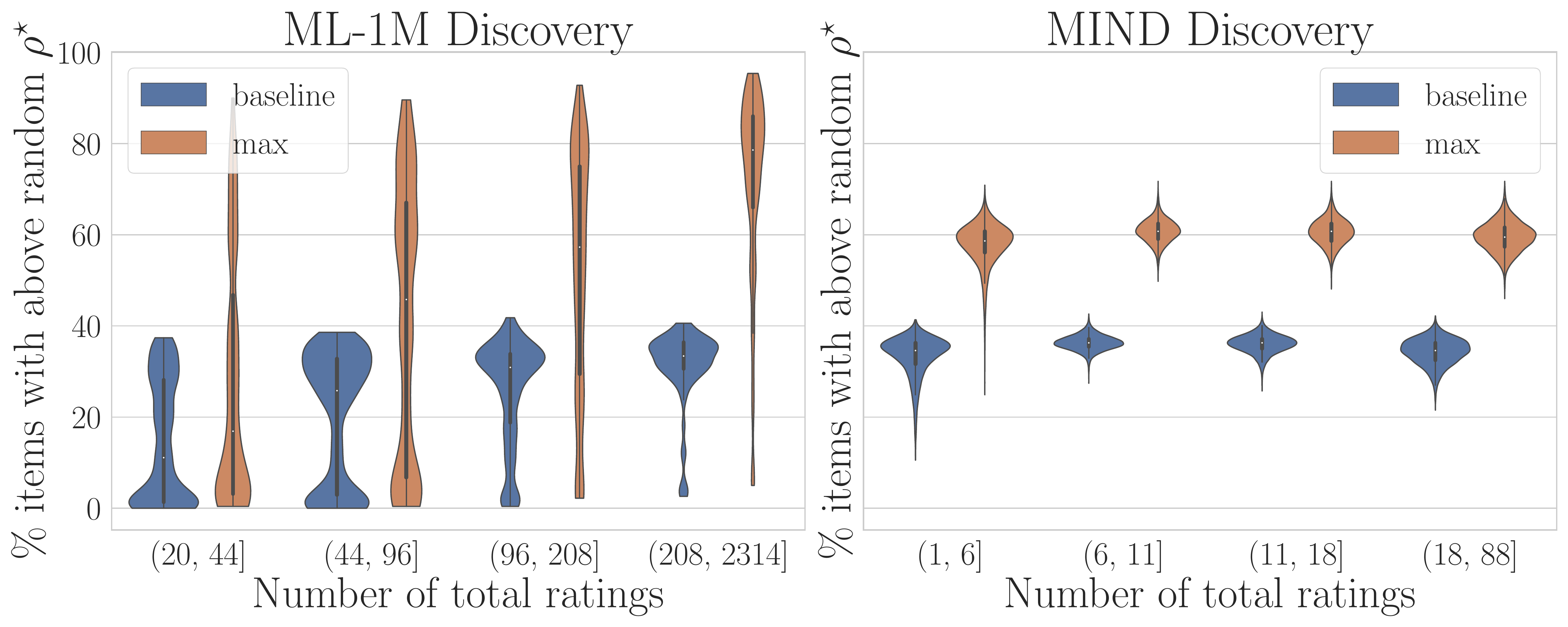}
    \caption{Comparison of baseline and best case fraction of items with better than random $\rho^\star$, grouped across four levels of user history length. Reachability evaluated on ML-1M (left) and MIND (right) for \emph{Next 10} action space, $\beta=2$, and LibFM model.} \label{fig:exp_bias}
\end{figure}

The reachability metric provides an upper bound for discovery and availability within a recommendation system.
While it has the benefit of making minimal assumptions about user behavior, the drawback is that it allows for perfectly strategic behaviors that would require users to have full knowledge of the internal structure of the model.
The results of a reachability audit may not be reflective of probable user experience, and thus reachability acts as a necessary but not sufficient condition.
{Nonetheless, reachability audit can lead to actionable insights by identifying inherent limits in system design. They allow} system designers to assess potential biases before releasing algorithmic updates into production. Moreover, as reachability depends on the choice of action space, such system-level insights might motivate user interface design: for example, a sidebar encouraging users to re-rate $K$ items from their history.

We point to a few directions of interest for future work.
Our result on the lack of trade-off between accuracy and reachability is encouraging.
Minimum one-step reachability conditions could be efficiently incorporated into learning algorithms for preference models.
It would also be interesting to extend reachability analysis to multiple interactions and longer time horizons.

Lastly, we highlight that the reachability lens presents a contrasting view to the popular line of work on robustness in machine learning.
When human behaviors are the subject of classification and prediction, building ``robustness'' into a system may be at odds with ensuring agency.
Because the goal of recommendation is personalization more than generalization, it would be appropriate to consider robust access over robust accuracy.
This calls for questioning the current normative stance and critically examining system desiderata in light of usage context.
 \section*{Acknowledgements}
This research is generously supported in part by ONR awards N00014-20-1-2497 and N00014-18-1-2833, NSF CPS award 1931853, and the DARPA Assured Autonomy program (FA8750-18-C-0101).
SD is supported by NSF GRF under Grant No. DGE 1752814.

\bibliographystyle{plainnat}
\bibliography{../icml2021/stochastic_reachability}

\begin{thebibliography}{45}
\providecommand{\natexlab}[1]{#1}
\providecommand{\url}[1]{\texttt{#1}}
\expandafter\ifx\csname urlstyle\endcsname\relax
  \providecommand{\doi}[1]{doi: #1}\else
  \providecommand{\doi}{doi: \begingroup \urlstyle{rm}\Url}\fi

\bibitem[Abdollahpouri et~al.(2019)Abdollahpouri, Mansoury, Burke, and
  Mobasher]{abdollahpouri2019impact}
Himan Abdollahpouri, Masoud Mansoury, Robin Burke, and Bamshad Mobasher.
\newblock The impact of popularity bias on fairness and calibration in
  recommendation.
\newblock \emph{arXiv preprint arXiv:1910.05755}, 2019.

\bibitem[ApS(2019)]{mosek}
MOSEK ApS.
\newblock \emph{MOSEK Optimizer API for Python Release 9.0.88}, 2019.
\newblock URL \url{https://docs.mosek.com/9.0/pythonapi.pdf}.

\bibitem[Castells et~al.(2011)Castells, Vargas, and Wang]{castells2011novelty}
Pablo Castells, Sa{\'u}l Vargas, and Jun Wang.
\newblock Novelty and diversity metrics for recommender systems: choice,
  discovery and relevance.
\newblock 2011.

\bibitem[Celma(2010)]{celma2010music}
Oscar Celma.
\newblock Music recommendation.
\newblock In \emph{Music recommendation and discovery}, pages 43--85. Springer,
  2010.

\bibitem[Cesa-Bianchi et~al.(2017)Cesa-Bianchi, Gentile, Lugosi, and
  Neu]{cesa2017boltzmann}
Nicol{\`o} Cesa-Bianchi, Claudio Gentile, G{\'a}bor Lugosi, and Gergely Neu.
\newblock Boltzmann exploration done right.
\newblock \emph{arXiv preprint arXiv:1705.10257}, 2017.

\bibitem[Chaney et~al.(2018)Chaney, Stewart, and
  Engelhardt]{chaney2018algorithmic}
Allison~JB Chaney, Brandon~M Stewart, and Barbara~E Engelhardt.
\newblock How algorithmic confounding in recommendation systems increases
  homogeneity and decreases utility.
\newblock In \emph{Proceedings of the 12th ACM Conference on Recommender
  Systems}, pages 224--232, 2018.

\bibitem[Chen et~al.(2020)Chen, Dong, Wang, Feng, Wang, and He]{chen2020bias}
Jiawei Chen, Hande Dong, Xiang Wang, Fuli Feng, Meng Wang, and Xiangnan He.
\newblock Bias and debias in recommender system: A survey and future
  directions.
\newblock \emph{arXiv preprint arXiv:2010.03240}, 2020.

\bibitem[Chen et~al.(2019)Chen, Beutel, Covington, Jain, Belletti, and
  Chi]{chen2019top}
Minmin Chen, Alex Beutel, Paul Covington, Sagar Jain, Francois Belletti, and
  Ed~H Chi.
\newblock Top-k off-policy correction for a reinforce recommender system.
\newblock In \emph{Proceedings of the Twelfth ACM International Conference on
  Web Search and Data Mining}, pages 456--464, 2019.

\bibitem[Christoffel et~al.(2015)Christoffel, Paudel, Newell, and
  Bernstein]{christoffel2015blockbusters}
Fabian Christoffel, Bibek Paudel, Chris Newell, and Abraham Bernstein.
\newblock Blockbusters and wallflowers: Accurate, diverse, and scalable
  recommendations with random walks.
\newblock In \emph{Proceedings of the 9th ACM Conference on Recommender
  Systems}, pages 163--170, 2015.

\bibitem[Dacrema et~al.(2021)Dacrema, Boglio, Cremonesi, and
  Jannach]{dacrema2021troubling}
Maurizio~Ferrari Dacrema, Simone Boglio, Paolo Cremonesi, and Dietmar Jannach.
\newblock A troubling analysis of reproducibility and progress in recommender
  systems research.
\newblock \emph{ACM Transactions on Information Systems (TOIS)}, 39\penalty0
  (2):\penalty0 1--49, 2021.

\bibitem[Dandekar et~al.(2013)Dandekar, Goel, and
  Lee]{dandekar2013polarization}
Pranav Dandekar, Ashish Goel, and David Lee.
\newblock Biased assimilation, homophily, and the dynamics of polarization.
\newblock In \emph{Proceedings of the National Academy of Sciences}, pages
  5791--5796, 2013.

\bibitem[Dean et~al.(2020)Dean, Rich, and Recht]{dean2020recommendations}
Sarah Dean, Sarah Rich, and Benjamin Recht.
\newblock Recommendations and user agency: the reachability of
  collaboratively-filtered information.
\newblock In \emph{Proceedings of the 2020 Conference on Fairness,
  Accountability, and Transparency}, pages 436--445, 2020.

\bibitem[Desrosiers and Karypis(2011)]{desrosiers2011comprehensive}
Christian Desrosiers and George Karypis.
\newblock A comprehensive survey of neighborhood-based recommendation methods.
\newblock \emph{Recommender systems handbook}, pages 107--144, 2011.

\bibitem[Ekstrand et~al.(2018{\natexlab{a}})Ekstrand, Tian, Azpiazu, Ekstrand,
  Anuyah, McNeill, and Pera]{ekstrand2018all}
Michael~D Ekstrand, Mucun Tian, Ion~Madrazo Azpiazu, Jennifer~D Ekstrand,
  Oghenemaro Anuyah, David McNeill, and Maria~Soledad Pera.
\newblock All the cool kids, how do they fit in?: Popularity and demographic
  biases in recommender evaluation and effectiveness.
\newblock In \emph{Conference on Fairness, Accountability and Transparency},
  pages 172--186. PMLR, 2018{\natexlab{a}}.

\bibitem[Ekstrand et~al.(2018{\natexlab{b}})Ekstrand, Tian, Kazi, Mehrpouyan,
  and Kluver]{ekstrand2018exploring}
Michael~D Ekstrand, Mucun Tian, Mohammed R~Imran Kazi, Hoda Mehrpouyan, and
  Daniel Kluver.
\newblock Exploring author gender in book rating and recommendation.
\newblock In \emph{Proceedings of the 12th ACM conference on recommender
  systems}, pages 242--250, 2018{\natexlab{b}}.

\bibitem[Faddoul et~al.(2020)Faddoul, Chaslot, and
  Farid]{faddoul2020longitudinal}
Marc Faddoul, Guillaume Chaslot, and Hany Farid.
\newblock A longitudinal analysis of youtube's promotion of conspiracy videos.
\newblock \emph{arXiv preprint arXiv:2003.03318}, 2020.

\bibitem[Flaxman et~al.(2016)Flaxman, Goel, and Rao]{flaxman2016filter}
Seth Flaxman, Sharad Goel, and Justin~M Rao.
\newblock Filter bubbles, echo chambers, and online news consumption.
\newblock \emph{Public opinion quarterly}, 80\penalty0 (S1):\penalty0 298--320,
  2016.

\bibitem[Harper and Konstan(2015)]{harper2015movielens}
F~Maxwell Harper and Joseph~A Konstan.
\newblock The movielens datasets: History and context.
\newblock \emph{Acm transactions on interactive intelligent systems (tiis)},
  5\penalty0 (4):\penalty0 1--19, 2015.

\bibitem[Harper et~al.(2015)Harper, Xu, Kaur, Condiff, Chang, and
  Terveen]{harper2015putting}
F~Maxwell Harper, Funing Xu, Harmanpreet Kaur, Kyle Condiff, Shuo Chang, and
  Loren Terveen.
\newblock Putting users in control of their recommendations.
\newblock In \emph{Proceedings of the 9th ACM Conference on Recommender
  Systems}, pages 3--10, 2015.

\bibitem[Herlocker et~al.(2004)Herlocker, Konstan, Terveen, and
  Riedl]{herlocker}
Jonathan Herlocker, Joseph Konstan, Loren Terveen, and John Riedl.
\newblock Evaluating collaborative filtering recommender systems.
\newblock \emph{ACM transactions on information systems}, 22\penalty0
  (1):\penalty0 5--53, 2004.

\bibitem[Ie et~al.(2019)Ie, Jain, Wang, Narvekar, Agarwal, Wu, Cheng, Chandra,
  and Boutilier]{ie2019slateq}
Eugene Ie, Vihan Jain, Jing Wang, Sanmit Narvekar, Ritesh Agarwal, Rui Wu,
  Heng-Tze Cheng, Tushar Chandra, and Craig Boutilier.
\newblock Slateq: A tractable decomposition for reinforcement learning with
  recommendation sets.
\newblock 2019.

\bibitem[Jannach et~al.(2015)Jannach, Lerche, Kamehkhosh, and
  Jugovac]{jannach2015recommenders}
Dietmar Jannach, Lukas Lerche, Iman Kamehkhosh, and Michael Jugovac.
\newblock What recommenders recommend: an analysis of recommendation biases and
  possible countermeasures.
\newblock \emph{User Modeling and User-Adapted Interaction}, 25\penalty0
  (5):\penalty0 427--491, 2015.

\bibitem[Karimi et~al.(2020)Karimi, Barthe, Sch{\"o}lkopf, and
  Valera]{karimi2020survey}
Amir-Hossein Karimi, Gilles Barthe, Bernhard Sch{\"o}lkopf, and Isabel Valera.
\newblock A survey of algorithmic recourse: definitions, formulations,
  solutions, and prospects.
\newblock \emph{arXiv preprint arXiv:2010.04050}, 2020.

\bibitem[Kawale et~al.(2015)Kawale, Bui, Kveton, Tran-Thanh, and
  Chawla]{kawale2015efficient}
Jaya Kawale, Hung~H Bui, Branislav Kveton, Long Tran-Thanh, and Sanjay Chawla.
\newblock Efficient thompson sampling for online matrix-factorization
  recommendation.
\newblock In \emph{Advances in neural information processing systems}, pages
  1297--1305, 2015.

\bibitem[Koren(2008)]{koren2008factorization}
Yehuda Koren.
\newblock Factorization meets the neighborhood: a multifaceted collaborative
  filtering model.
\newblock In \emph{Proceedings of the 14th ACM SIGKDD international conference
  on Knowledge discovery and data mining}, pages 426--434, 2008.

\bibitem[Koren and Bell(2015)]{koren2015advances}
Yehuda Koren and Robert Bell.
\newblock Advances in collaborative filtering.
\newblock \emph{Recommender systems handbook}, pages 77--118, 2015.

\bibitem[Krauth et~al.(2020)Krauth, Dean, Zhao, Guo, Curmei, Recht, and
  Jordan]{krauth2020offline}
Karl Krauth, Sarah Dean, Alex Zhao, Wenshuo Guo, Mihaela Curmei, Benjamin
  Recht, and Michael~I Jordan.
\newblock Do offline metrics predict online performance in recommender systems?
\newblock \emph{arXiv preprint arXiv:2011.07931}, 2020.

\bibitem[Lukoff et~al.(2021)Lukoff, Lyngs, Zade, Liao, Choi, Fan, Munson, and
  Hiniker]{lukoff2021design}
Kai Lukoff, Ulrik Lyngs, Himanshu Zade, J~Vera Liao, James Choi, Kaiyue Fan,
  Sean~A Munson, and Alexis Hiniker.
\newblock How the design of youtube influences user sense of agency.
\newblock \emph{arXiv preprint arXiv:2101.11778}, 2021.

\bibitem[Mary et~al.(2015)Mary, Gaudel, and Preux]{mary2015bandits}
J{\'e}r{\'e}mie Mary, Romaric Gaudel, and Philippe Preux.
\newblock Bandits and recommender systems.
\newblock In \emph{International Workshop on Machine Learning, Optimization and
  Big Data}, pages 325--336. Springer, 2015.

\bibitem[Nguyen et~al.(2014)Nguyen, Hui, Harper, Terveen, and
  Konstan]{nguyen2014exploring}
Tien~T Nguyen, Pik-Mai Hui, F~Maxwell Harper, Loren Terveen, and Joseph~A
  Konstan.
\newblock Exploring the filter bubble: the effect of using recommender systems
  on content diversity.
\newblock In \emph{Proceedings of the 23rd international conference on World
  wide web}, pages 677--686, 2014.

\bibitem[Ning and Karypis(2011)]{ning2011slim}
Xia Ning and George Karypis.
\newblock Slim: Sparse linear methods for top-n recommender systems.
\newblock In \emph{2011 IEEE 11th International Conference on Data Mining},
  pages 497--506. IEEE, 2011.

\bibitem[Rendle(2012)]{rendle:tist2012}
Steffen Rendle.
\newblock Factorization machines with {libFM}.
\newblock \emph{ACM Trans. Intell. Syst. Technol.}, 3\penalty0 (3):\penalty0
  57:1--57:22, May 2012.
\newblock ISSN 2157-6904.

\bibitem[Rendle et~al.(2019)Rendle, Zhang, and Koren]{rendle2019difficulty}
Steffen Rendle, Li~Zhang, and Yehuda Koren.
\newblock On the difficulty of evaluating baselines: A study on recommender
  systems.
\newblock \emph{arXiv preprint arXiv:1905.01395}, 2019.

\bibitem[Ribeiro et~al.(2020)Ribeiro, Ottoni, West, Almeida, and
  Meira~Jr]{ribeiro2020auditing}
Manoel~Horta Ribeiro, Raphael Ottoni, Robert West, Virg{\'\i}lio~AF Almeida,
  and Wagner Meira~Jr.
\newblock Auditing radicalization pathways on youtube.
\newblock In \emph{Proceedings of the 2020 conference on fairness,
  accountability, and transparency}, pages 131--141, 2020.

\bibitem[Schnabel et~al.(2016)Schnabel, Swaminathan, Singh, Chandak, and
  Joachims]{schnabel2016recommendations}
Tobias Schnabel, Adith Swaminathan, Ashudeep Singh, Navin Chandak, and Thorsten
  Joachims.
\newblock Recommendations as treatments: Debiasing learning and evaluation.
\newblock In \emph{international conference on machine learning}, pages
  1670--1679. PMLR, 2016.

\bibitem[Shakespeare et~al.(2020)Shakespeare, Porcaro, G{\'o}mez, and
  Castillo]{shakespeare2020exploring}
Dougal Shakespeare, Lorenzo Porcaro, Emilia G{\'o}mez, and Carlos Castillo.
\newblock Exploring artist gender bias in music recommendation.
\newblock \emph{arXiv preprint arXiv:2009.01715}, 2020.

\bibitem[Singh and Joachims(2018)]{singh2018fairness}
Ashudeep Singh and Thorsten Joachims.
\newblock Fairness of exposure in rankings.
\newblock In \emph{Proceedings of the 24th ACM SIGKDD International Conference
  on Knowledge Discovery \& Data Mining}, pages 2219--2228, 2018.

\bibitem[Steck(2018)]{steck2018calibrated}
Harald Steck.
\newblock Calibrated recommendations.
\newblock In \emph{Proceedings of the 12th ACM conference on recommender
  systems}, pages 154--162, 2018.

\bibitem[Steck(2019)]{steck2019embarrassingly}
Harald Steck.
\newblock Embarrassingly shallow autoencoders for sparse data.
\newblock In \emph{The World Wide Web Conference}, pages 3251--3257, 2019.

\bibitem[Ustun et~al.(2019)Ustun, Spangher, and Liu]{ustun2019actionable}
Berk Ustun, Alexander Spangher, and Yang Liu.
\newblock Actionable recourse in linear classification.
\newblock In \emph{Proceedings of the Conference on Fairness, Accountability,
  and Transparency}, pages 10--19, 2019.

\bibitem[Wei et~al.(2017)Wei, Xu, Lan, Guo, and Cheng]{wei2017reinforcement}
Zeng Wei, Jun Xu, Yanyan Lan, Jiafeng Guo, and Xueqi Cheng.
\newblock Reinforcement learning to rank with markov decision process.
\newblock In \emph{Proceedings of the 40th International ACM SIGIR Conference
  on Research and Development in Information Retrieval}, pages 945--948, 2017.

\bibitem[Wu et~al.(2020)Wu, Qiao, Chen, Wu, Qi, Lian, Liu, Xie, Gao, Wu,
  et~al.]{wu2020mind}
Fangzhao Wu, Ying Qiao, Jiun-Hung Chen, Chuhan Wu, Tao Qi, Jianxun Lian,
  Danyang Liu, Xing Xie, Jianfeng Gao, Winnie Wu, et~al.
\newblock Mind: A large-scale dataset for news recommendation.
\newblock In \emph{Proceedings of the 58th Annual Meeting of the Association
  for Computational Linguistics}, pages 3597--3606, 2020.

\bibitem[Yang et~al.(2018)Yang, Cui, Xuan, Wang, Belongie, and
  Estrin]{yang2018unbiased}
Longqi Yang, Yin Cui, Yuan Xuan, Chenyang Wang, Serge Belongie, and Deborah
  Estrin.
\newblock Unbiased offline recommender evaluation for missing-not-at-random
  implicit feedback.
\newblock In \emph{Proceedings of the 12th ACM Conference on Recommender
  Systems}, pages 279--287, 2018.

\bibitem[Yao et~al.(2021)Yao, Halpern, Thain, Wang, Lee, Prost, Chi, Chen, and
  Beutel]{yao2021measuring}
Sirui Yao, Yoni Halpern, Nithum Thain, Xuezhi Wang, Kang Lee, Flavien Prost,
  Ed~H Chi, Jilin Chen, and Alex Beutel.
\newblock Measuring recommender system effects with simulated users.
\newblock \emph{arXiv preprint arXiv:2101.04526}, 2021.

\bibitem[Zhou et~al.(2008)Zhou, Wilkinson, Schreiber, and Pan]{zhou2008large}
Yunhong Zhou, Dennis Wilkinson, Robert Schreiber, and Rong Pan.
\newblock Large-scale parallel collaborative filtering for the netflix prize.
\newblock In \emph{International conference on algorithmic applications in
  management}, pages 337--348. Springer, 2008.

\end{thebibliography}
\appendix

\newpage
\onecolumn
\appendix

\section{Further Examples} \label{app:examples}

\begin{example}[Biased MF-SGD]
Biased matrix factorization models~\cite{koren2015advances} compute scores as rating predictions with
\[s_{ui} = \vec p_u^\top \vec q_i + f_u + g_i + \mu\]
 $P\in\mathbb{R}^{\nusers\times \latentdim}$ and $Q\in\mathbb{R}^{\nitems\times \latentdim}$ are respectively user and item factors for some latent dimension $\latentdim$, $\vec f\in\mathbb{R}^\nusers$ and $\vec g\in\mathbb{R}^\nitems$ are respectively user and item biases, and $\mu\in\mathbb{R}$ is a global bias. 

 The parameters are learned via the regularized optimization
\[\min_{P,Q,\vec f,\vec g,\mu} \frac{1}{2}\sum_{u}\sum_{i\in\Omega_u}\|\vec p_u^\top \vec q_i+ f_u + g_i + \mu - r_{ui}\|_2^2 + \frac{\lambda}{2} \|P\|_F^2 + \frac{\lambda}{2}\|Q\|_F^2 \:.\]
Under a stochastic gradient descent minimization scheme~\cite{koren2008factorization} with step size $\alpha$, the one-step
update rule for a user factor is
\[\vec p^+_u = \vec p_u - \alpha \sum_{i\in\Omega_u^\calA}(\vec q_i\vec q_i^\top \vec p_u + \vec q_i(f_u + g_i + \mu) - \vec q_i r_{ui}) -\alpha \lambda\vec p_u\:.\]
User bias terms can be updated in a similar manner, but because the user bias is equal across items, it does not impact the selection of items.

Notice that this expression is affine in the mutable ratings.
Therefore, we have an affine score function:
\[\phi_u(\vec a) = Q\vec p_u^+=Q \left( (1-\alpha\lambda)\vec p_u - \alpha Q_\calA^\top (Q_\calA \vec p_u+\vec g_{\cal A}+(\mu+f_u) \mathbf{1}) + \alpha Q_\calA^\top \vec a \right)\]
where we define $Q_\calA = Q_{\Omega_u^\calA}\in \mathbb{R}^{|\Omega_u^\calA|\times d}$ and $\vec g_\calA = \vec g_{\Omega_u^\calA} \in\mathbb{R}^{|\Omega_u^\calA|}$.
Therefore,
\[B_u = \alpha Q Q_\calA^\top,\quad
\vec c_u = Q \left( (1+\lambda)\vec p_u - \alpha Q_\calA^\top (Q_\calA \vec p_u+\vec g_{\cal A}+(\mu+f_u) \mathbf{1})\right)\:.\]
\end{example}

\begin{example}[Biased MF-ALS]
Rather than a stochastic gradient descent minimization scheme, we may instead update the model with an alternating least-squares strategy~\cite{zhou2008large}.
In this case, the update rule is
\begin{align*}
\vec p^+_u &= \arg\min_{\vec p} \sum_{i\in\Omega_u^\calA\cap\Omega_u}\|\vec p^\top \vec q_i+ f_u + g_i + \mu - r_{ui}\|_2^2 + \lambda \|\vec p\|_2^2\\
&= (Q_u^\top Q_u + \lambda I)^{-1} (Q^\top \vec r_u + Q_\calA^\top (\vec g_\calA +(\mu+f_u) \mathbf{1}) + Q_\calA^\top \vec a )
\end{align*}
where we define $Q_u = Q_{\Omega_u^\calA\cap\Omega_u}$.
Similar to in the SGD setting, this is an affine expression, and therefore we end up with the affine score parameters
\[B_u = Q(Q_u^\top Q_u + \lambda I)^{-1} Q_\calA^\top,\quad
\vec c_u = Q (Q_u^\top Q_u + \lambda I)^{-1}(Q^\top \vec r_u + Q_\calA^\top (\vec g_\calA +(\mu+f_u) \mathbf{1}))\:.\]
\end{example}

\begin{example}[Biased Item-KNN]
Biased neighborhood models~\cite{desrosiers2011comprehensive} compute scores as rating predictions by a weighted average, with
\[s_{ui} = \mu+f_u + g_i +\frac{\sum_{j\in\mathcal N_i} w_{ij} (r_{uj} -\mu-f_u-g_i)}{\sum_{j\in\mathcal N_i} |w_{ij}|}\]
where $w_{ij}$ are weights representing similarities between items, $\mathcal N_i$ is a set of indices which are in the neighborhood of item $i$, and $\vec f,\vec g,\mu$ are bias terms. Regardless of the details of how these parameters are computed, the predicted scores are an affine function of observed scores: 
\[\vec s_u = W \vec r_u - W(\vec g + (\mu+f_u)\vec 1) + \vec g + (\mu+f_u)\vec 1 \]
 where we can define
\[W_{ij} = \begin{cases}\frac{w_{ij}}{\sum_{j\in\mathcal N_i} |w_{ij}|} & j\in\mathcal N_i\\ 0 & \text{otherwise} \end{cases}\]

Therefore, the score updates take the form
\[\phi_u(\vec a) = \underbrace{W(\vec r_u -\vec g + (\mu+f_u)\vec 1) + \vec g + (\mu+f_u)\vec 1 )}_{\vec c_u} + \underbrace{W E_{\Omega_u^\calA}}_{B_u} \vec a\]
where $E_{\Omega_u^\calA}$ selects rows of $W$ corresponding to action items.
\end{example}

\begin{example}[SLIM and EASE]
For both SLIM~\cite{ning2011slim} and EASE~\cite{steck2019embarrassingly}, scores are computed as
\[s_{ui} = \vec w_i^\top \r_u\]
for $\vec w_i$ the row vectors of a weight matrix $W$.
For SLIM, the sparse weights are computed as
\begin{align*}
\min_W~&\frac{1}{2}\|R - RW\|_F^2 + \frac{\beta}{2}\|W\|_F^2 + \lambda\|W\|_1\\
\text{s.t.}~&W\geq 0,\mathrm{diag}(W)=0
\end{align*}
For EASE, the weights are computed as
\begin{align*}
\min_W~&\frac{1}{2}\|R - RW\|_F^2 + \lambda\|W\|_F^2\\
\text{s.t.}~&\mathrm{diag}(W)=0
\end{align*}

In both cases, the score updates take the form
\[\phi_u(\vec a) = \underbrace{W\vec r_u}_{\vec c_u} + \underbrace{W E_{\Omega_u^\calA}}_{B_u} \vec a\:.\]
\end{example}

\section{Datasets, Model Training and Computing Infrastructure}\label{app:data_recs}

\subsection{Detailed data description}\label{app:data_details}

\paragraph{MovieLens 1 Million}
ML-1M dataset was downloaded from Group Lens\footnote{\url{https://grouplens.org/datasets/movielens/1m/}} via the RecLab~\cite{krauth2020offline} interface\footnote{\url{https://github.com/berkeley-reclab/RecLab}}. It contains 1 through 5 rating data of 6040 users for 3706 movies. There are a total of 1000209 ratings (4.47\% rating density). The original data is accompanied by additional user attributes such as age, gender, occupation and zip code. Our experiments didn't indicate observable biases across these attributes. In Section \ref{app:experiments} we show user discovery results split by gender.

Figure \ref{fig:ml1m_stats} illustrates descriptive statistics for the ML-1M dataset.
\begin{figure}
    \center
    \includegraphics[width=0.98\columnwidth]{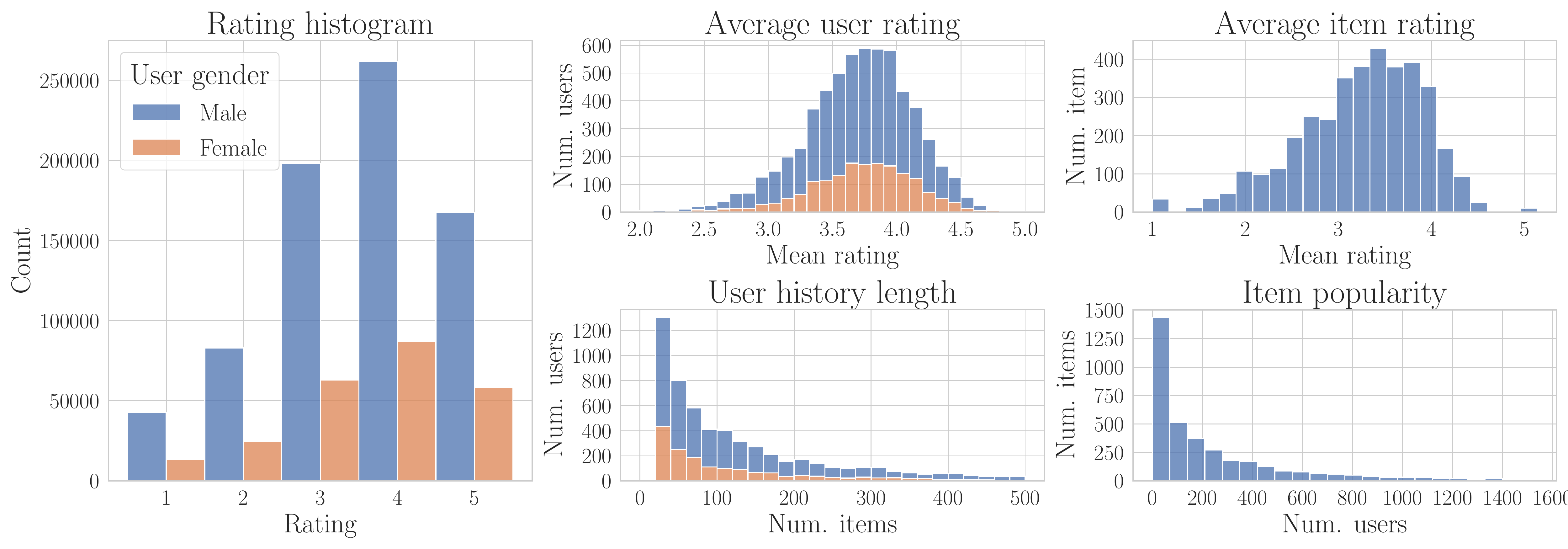}
    \caption{Descriptive statistics for the MovieLens 1M dataset split by user gender (28.3\% female).  The mean ratings of both users and items are roughly normally distributed while user's history length and item popularity display power law distributions.} \label{fig:ml1m_stats}
\end{figure}

\paragraph{LastFM 360K}
The LastFM 360K dataset preprocessed\footnote{\url{https://zenodo.org/record/3964506\#.XyE5N0FKg5n}} by~\citet{shakespeare2020exploring} was loaded via the RecLab interface. It contains data on the number of times users have listened to various artists. We select a random subset of 10\% users and a random subset of 10\% items yielding  13698 users, 20109 items and 178388 ratings (0.056\% rating density). The item ratings are not explicitly expressed by users as in the MovieLens case. For a user $u$ and an artist $i$ we define implicit ratings $r_{ui} = \log (\# \mathrm{listens}(u,i) + 1)$. This data is accompanied by artist gender, an item attribute.

Figure \ref{fig:lastfm_stats} illustrates descriptive statistics for the LastFM dataset.

\begin{figure}
    \center
    \includegraphics[width=0.98\columnwidth]{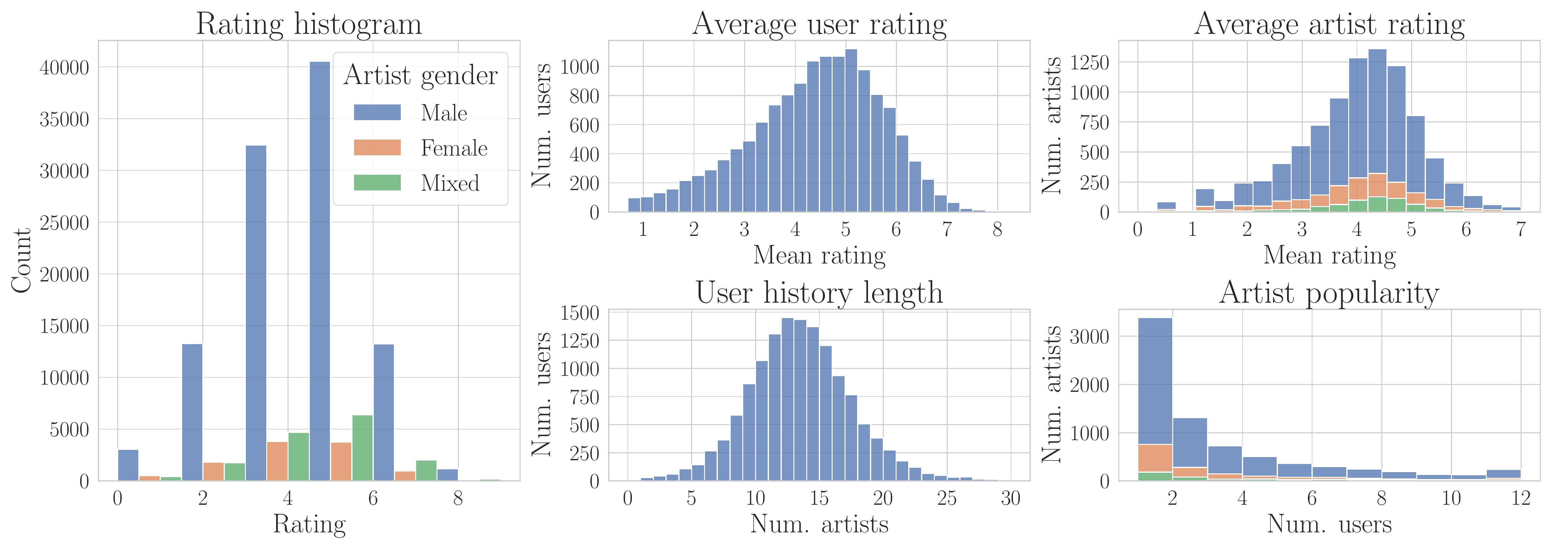}
    \caption{Descriptive statistics for the LastFM dataset split by artist gender (over 54\% of artists have unknown gender, 36\% are male, 6.5\% are female and 3.5\% are mixed gender). Unlike ML 1M, for the LastFM dataset the user history lengths are normally distributed around a mean of around 12 artists.} \label{fig:lastfm_stats}
\end{figure}

\paragraph{MIcrosoft News Dataset (MIND)}
MIND is a recently published impression dataset collected from logs of the Microsoft News website \footnote{\url{https://microsoftnews.msn.com/}}. We downloaded the \texttt{MIND-small} dataset\footnote{\url{https://msnews.github.io/}}, which contains behaviour log data for 50000 randomly sampled users. There are 42416 unique news articles, spanning 17 categories and 247 subcategories. We aggregate user interactions at the subcategory level and consider the problem of news subcategory recommendation. The implicit rating of a user $u$ for subcategory $i$ is defined as: $r_{ui} = \log (\# \mathrm{clicks}(u,i) + 1)$. The resulting aggregated dataset contains 670773 ratings (5.54\% rating density).

Figure \ref{fig:mind_stats} illustrates descriptive statistics for the MIND dataset.

\begin{figure}
    \center
    \includegraphics[width=0.98\columnwidth]{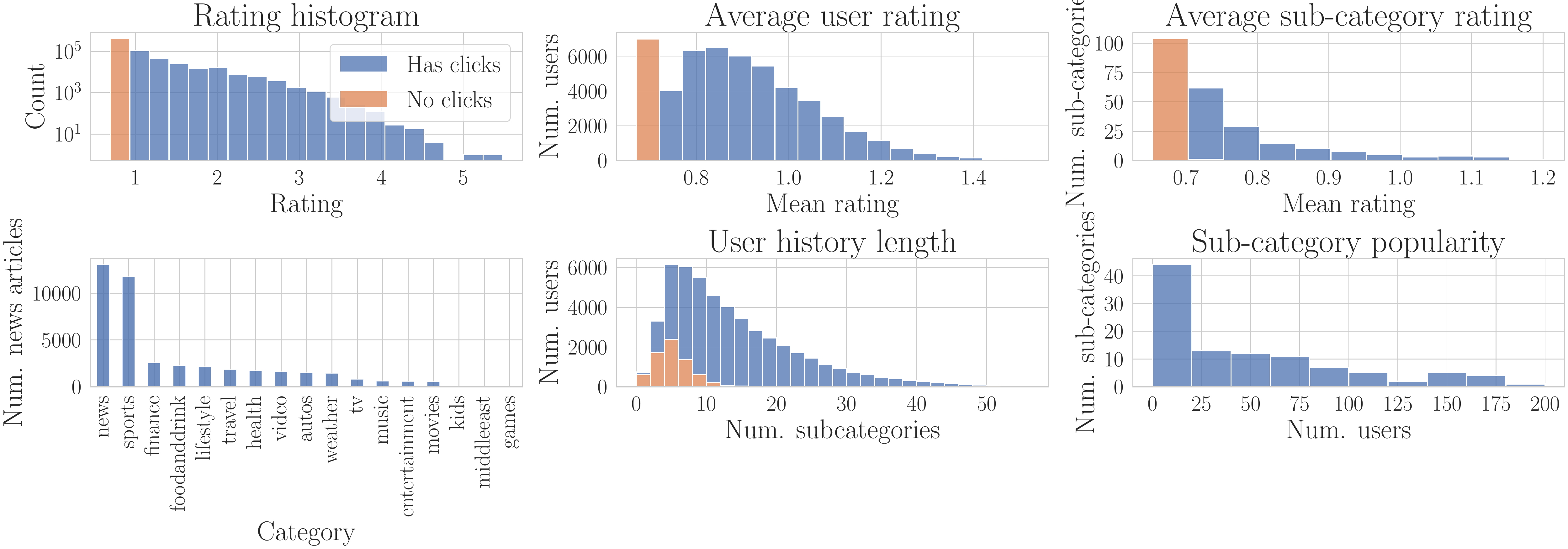}
    \caption{Descriptive statistics for the MIND dataset: The orange bars correspond to either user or items that have been displayed but have not clicked/ have not been clicked on. Unlike ML 1M and LastFM, the MIND ratings have strongly skewed distribution, with most user-subcategory ratings corresponding to users clicking on a small number of articles from the sub-category. There is a long tail of higher ratings that corresponds to most popular subcategories. The leftmost plot illustrates the unequal distribution of news articles across categories. The same qualitative behaviour holds for sub-categories.} \label{fig:mind_stats}
\end{figure}

\subsection{Model Tuning}\label{app:model_tuning}
For each dataset and recommender model we perform grid search for progressively finer meshes over the tunable hyper-parameters of the recommender.
We use recommenders implemented by the RecLab library.
For each dataset and recommender we evaluate hyperparameters on a 10\% split of test data. The best hyper-parameters for each setting are presented in Table \ref{tbl:tuning}.

\paragraph{LibFM}
We performed hyper-parameter tuning to find suitable learning rate and regularization parameter for each dataset. Following~\cite{dacrema2021troubling} we consider $\texttt{lr} \in(0.001, 0.5)$ as the range of hyper-parameters for the learning rate and $\texttt{reg}\in (10^{-5}, 10^{0})$ for the regularization parameter. In all experimental settings we follow the setup of~\cite{rendle2019difficulty} and use 64 latent dimensions and train with SGD for 128 iterations.

\paragraph{KNN}
We perform hyperparameter tuning with respect to neighborhood size and shrinkage parameter. Following~\cite{dacrema2021troubling} we consider the range $(5, 1000)$ for the neighborhood size and $(0,1000)$ for the shrinkage parameter. We tune KNN only for the ML-1M dataset.

\begin{table}[t]
    \caption{Tuning results}
    \label{tbl:tuning}
    \vskip 0.15in
    \begin{center}
    \begin{small}
    \begin{sc}
    \begin{tabular}{l|cccc|}
    \toprule
     & \multicolumn{4}{c|}{LibFM} \\
     Dataset & LR & Reg. & Test RMSE & Run time (s)  \\
    \midrule
    ML 1M   &  0.0112& 0.0681 & 0.716 &2.76 $\pm$ 0.32 \\
    LastFM     & 0.0478 & 0.2278 &1.122 & 0.78 $\pm$ 0.13\\
    MIND     & 0.09& 0.0373& 0.318 & 3.23 $\pm$ 0.37\\
    \midrule
     &  \multicolumn{4}{c|}{KNN} \\
     Dataset & Neigh. size & Shrinkage & Test RMSE & Run time (s)\\
    \midrule
    ML 1M  & 100 & 22.22& 0.756 & 0.34 $\pm$ 0.07\\
    \bottomrule
    \end{tabular}
    \end{sc}
    \end{small}
    \end{center}
    \vskip -0.1in
    \end{table}

\subsection{Experimental Infrastructure and Computational Complexity}\label{app:reach_complexity}
All experiments were performed on a 64 bit desktop machine equipped with 20 CPUs (Intel(R) Core(TM) i9-7900X CPU @ 3.30GHz) and a 62 GiB RAM.  Average run times for training an instance of each recommender can be found in Table \ref{tbl:tuning}.

\section{Computing Reachability}

\subsection{Conic Program Implementation}

The optimization problem in~\eqref{eq:linear_recourse} is convex, and we solve it as a conic optimization problem using the MOSEK Python API under an academic license~\cite{mosek}.
We reformulate~\eqref{eq:linear_recourse} as an optimization over the exponential cone:
\begin{align}
\begin{split} \label{eq:exp_cone_opt}
\min_{\substack{t,\vec a,\vec u}}~~& t -\beta (\vec b_{ui}^\top \vec a + c_{ui})\\
\text{s.t.}~~&\vec a\in\calA_u,\quad \sum_{j\in\Omega_u^t} u_j \leq 1,\\
&\left(u_j, 1, \beta(\vec b_{uj}^\top \vec a + c_{uj})-t\right) \in\mathcal K_{exp}~~\forall~~j\in\Omega_u^t
\end{split}
\end{align}

The parameters $B_u$ and $\vec c_u$ are computed for each user based on the recommender model as described in Section~\ref{app:examples}.
For the LibFM model, we consider user updates with $\alpha=0.1$ and $\lambda=0$.
Average run times for computing reachability of a user-item pair in various settings can be found in Table \ref{tbl:run_times}.

\subsection{Experimental Setup for Computing Reachability}
\paragraph{ML 1M} We compute max stochastic reachability for the LibFM and KNN preference model. We consider three types of user action spaces: \emph{History Edits}, \emph{Future Edits}, and \emph{Next K} in which users can strategically modify the ratings associated to $K$ randomly chosen items from their history, $K$ randomly chosen items from that they have not yet seen, or the top-$K$ unseen items according to the baseline scores of the preference model. For each of the action spaces we consider
$K\in\{5,10,20\}$.

We perform reachability experiments on a random 3\% subset of users (176).
For each choice of preference model, action space type and action space size we sample for each user 500 random items that have not been previously rated and are not action items. For each user-item pair we compute reachability for a range of stochasticity parameters $\beta  \in \{1,2,4,10\}$. Note that across all experimental settings we compute reachability for the same subset of users, but different subsets of randomly selected target items.

We use the ML 1M dataset to primarily gain insights in the role that preference models, item selection stochasticity and strategic action spaces play in determining the maximum achievable degree of stochastic reachability in a recommender system.

\paragraph{LastFM}
We run reachability experiment for LibFM recommender with \emph{Next K = 10} action model and stochasticity parameter $\beta = 2$. We compute $\rho^\star$ values for 100 randomly sampled users and 500 randomly sampled items from the set of non-action items (target items can include previously seen items). Unlike the ML 1M dataset, the set of target items is shared among all users.

\paragraph{MIND}
We run reachability experiments for LibFM recommender with \emph{Next K = 10} action model and stochasticity parameter $\beta = 2$. We compute reachability for all items and users.

\paragraph{Reachability Run Times}
In Table \ref{tbl:run_times} we present the average clock time for computing reachability for a user-item pair in the settings described above. Due to internal representation of action spaces as matrices the runtime dependence on the dimension of the action space is fairly modest. We do not observe significant run time differences between different types of action spaces. We further add multiprocessing functionality to parallelize reachability computations over multiple target items.

\begin{table}[t]
    \caption{Reachability run times (in seconds).}
    \label{tbl:run_times}
    \vskip 0.15in
    \begin{center}
    \begin{small}
    \begin{sc}
    \begin{tabular}{l|c|c|c|r}
    \toprule
    Num. actions &{ML 1M (LibFM)} & {ML 1M (KNN)} & {LastFM} &  {MIND}  \\
    \midrule
    K = 5   & 0.82 $\pm$ 0.04& 9.8 $\pm$ 3.4 & -&-\\
    K = 10 & 0.87 $\pm$ 0.04 & 10.2 $\pm$ 6.1 & 4.91 $\pm$ 0.32 & 0.44  $\pm$ 0.01  \\
    K= 20  & 0.91 $\pm$ 0.05 & 11.4 $\pm$ 6.8 &-&-\\
    \bottomrule
    \end{tabular}
    \end{sc}
    \end{small}
    \end{center}
    \vskip -0.1in
    \end{table}

\section{Detailed Experimental Results} \label{app:experiments}
\subsection{Impact of recommender design}

We present further insights in the experimental settings studied in Section \ref{sec:pipeline}. For ML-1M, we replicate the log scale scatterplots of $\rho^\star$ against baseline $\rho$ for all the action spaces (\emph{Next K, Random Future, Random History}), the full range of $\beta\in\{1,2,4,10\}$ and the two preference models: LibFM (Figure \ref{fig:libfm_full}) and KNN (Figure \ref{fig:knn_full}). We observe that for both KNN and LibFM, random history edits can lead to higher $\rho^\star$ values. We posit that this increased agency is partly due to the fact that when editing $K$ items from the history a user edits a larger fraction of total ratings compared to editing $K$ future items.

The most striking feature of KNN reachability results is the strong correlation between baseline $\rho$ and $\rho^\star$. The correlations between baseline and max probability of recommendation is less strong in the case of LibFM. These insights are corroborated by Figure \ref{fig:violin_full} which compares the average LibFM and KNN user lifts for different choices of action space, action size \emph{K}, stochasticity parameter $\beta$.

\begin{figure}
    \center
    \includegraphics[width=0.98\columnwidth]{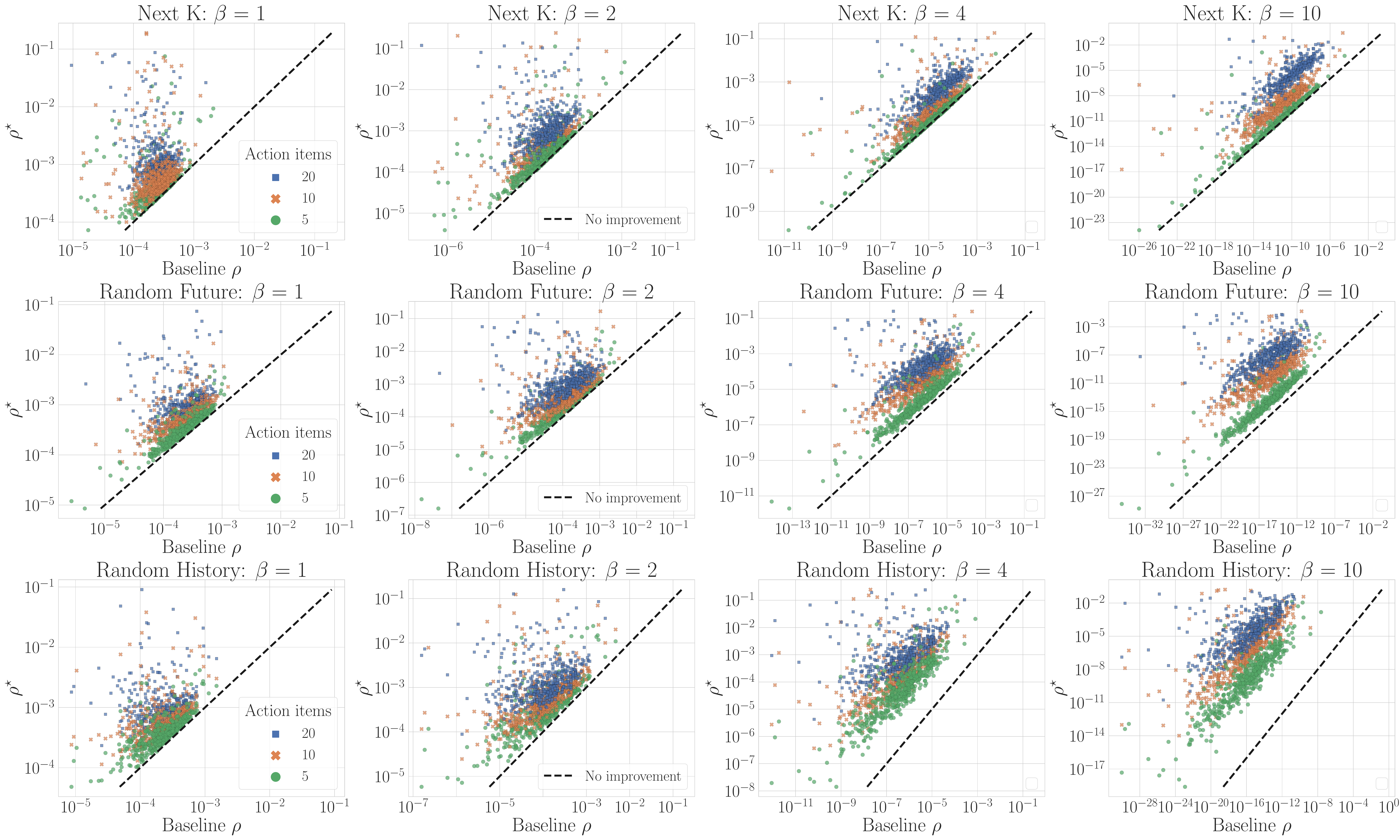}
    \caption{Log scale scatterplots of $\rho^\star$ against baseline $\rho$ evaluated for the LibFM preference model.} \label{fig:libfm_full}
\end{figure}
\begin{figure}
    \center
    \includegraphics[width=0.98\columnwidth]{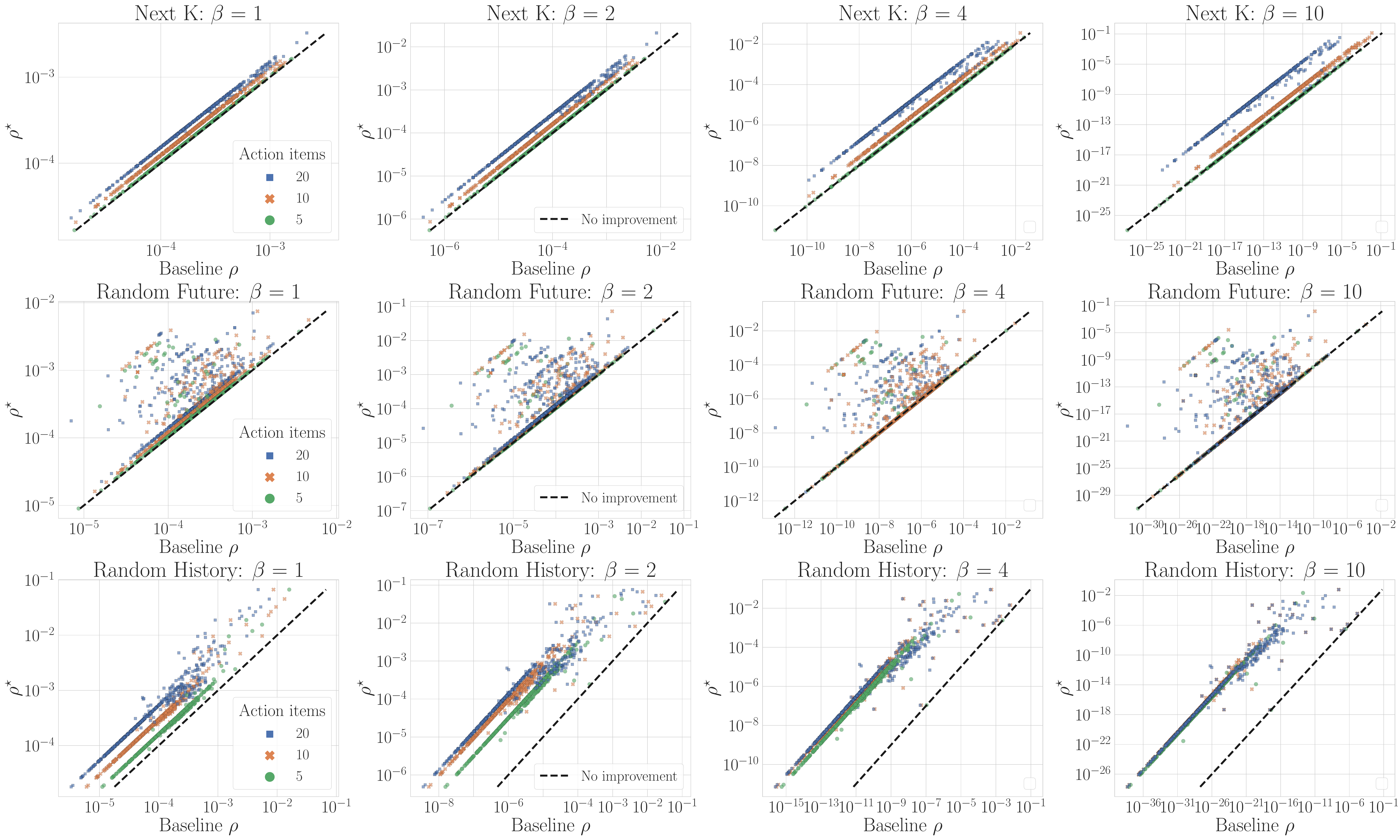}
    \caption{Log scale scatterplots of $\rho^\star$ against baseline $\rho$ evaluated for the KNN preference model.} \label{fig:knn_full}
\end{figure}
\begin{figure}
    \center
    \includegraphics[width=0.98\columnwidth]{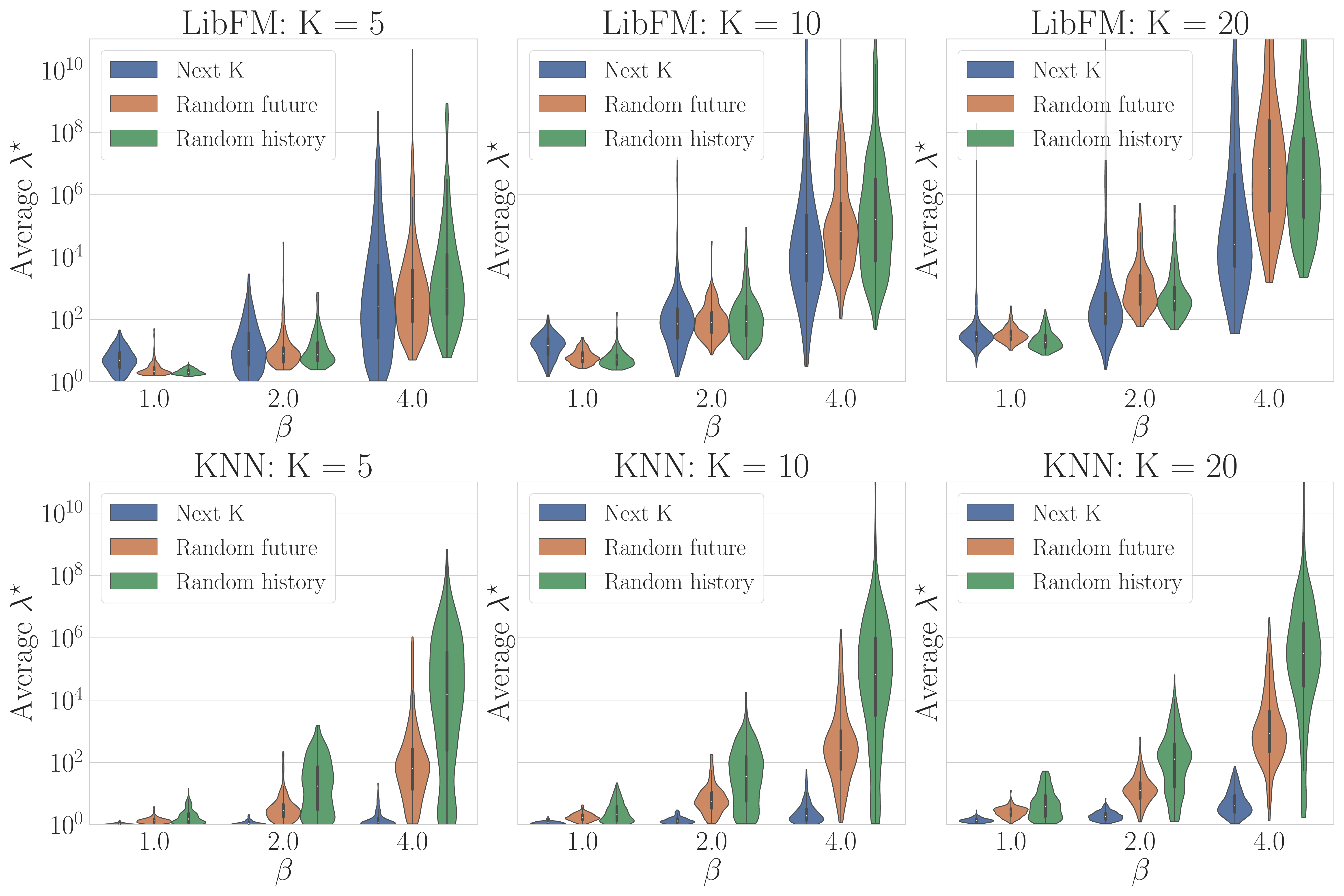}
    \caption{Side by side comparison of average user lifts for LibFM (top row)and KNN(bottom row).} \label{fig:violin_full}
\end{figure}

\subsection{Bias in movie, music, and news recommendation}
\begin{figure}
    \center
    \includegraphics[width=0.98\columnwidth]{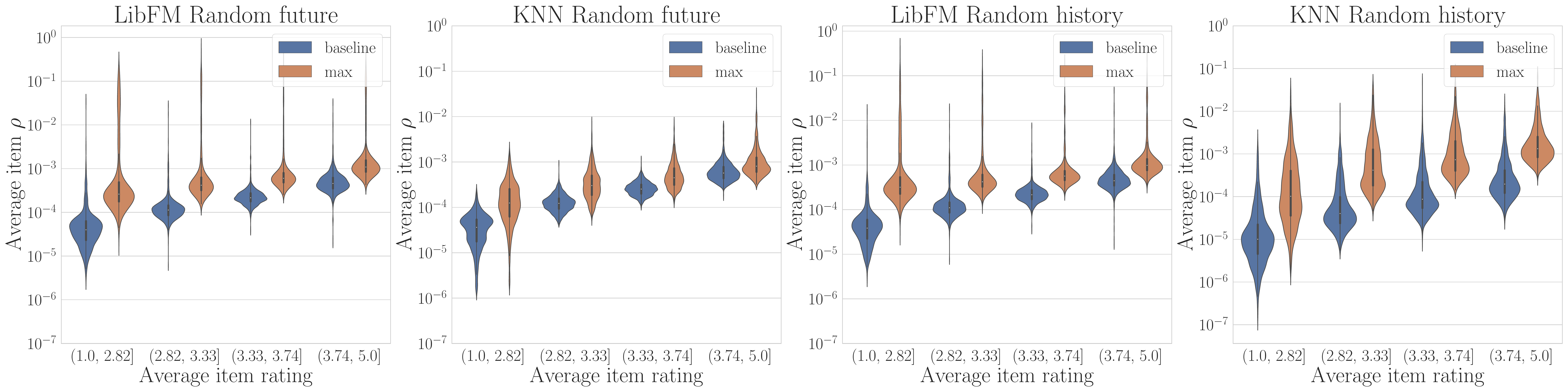}
    \caption{Side by side comparison of baseline and best-case availability of content, across four popularity categories. From left to right: LibFM preference model with \emph{Random Future}, KNN preference model with \emph{Random Future}, LibFM preference model with \emph{Random History}, KNN preference model with \emph{Random History}. Reachability evaluated on ML-1M for with $K=10$ and $\beta=2$.} \label{fig:pop_bias_app}
\end{figure}

We present further results on the settings studied in Section~\ref{sec:exp_bias}.
We replicate the popularity bias results on ML-1M for different action spaces and plot the results in Figure~\ref{fig:pop_bias_app}.
We see that the availability bias for KNN is dependent on the action space, with \emph{Random History} displaying no or little correlation between popularity and max availability.
This is not surprising given the results in Figure~\ref{fig:knn_libfm_lift}.

To systematically study the popularity bias, we compute the Spearman rank-order correlation coefficient to measure the presence of a monotonic relationship between popularity (as measured by average rating) and availability (either in the baseline or max case). We also compute the correlation between the popularity and the prevalence in the dataset, as measured by number of ratings.

\begin{figure}
    \center
    \includegraphics[width=0.6\columnwidth]{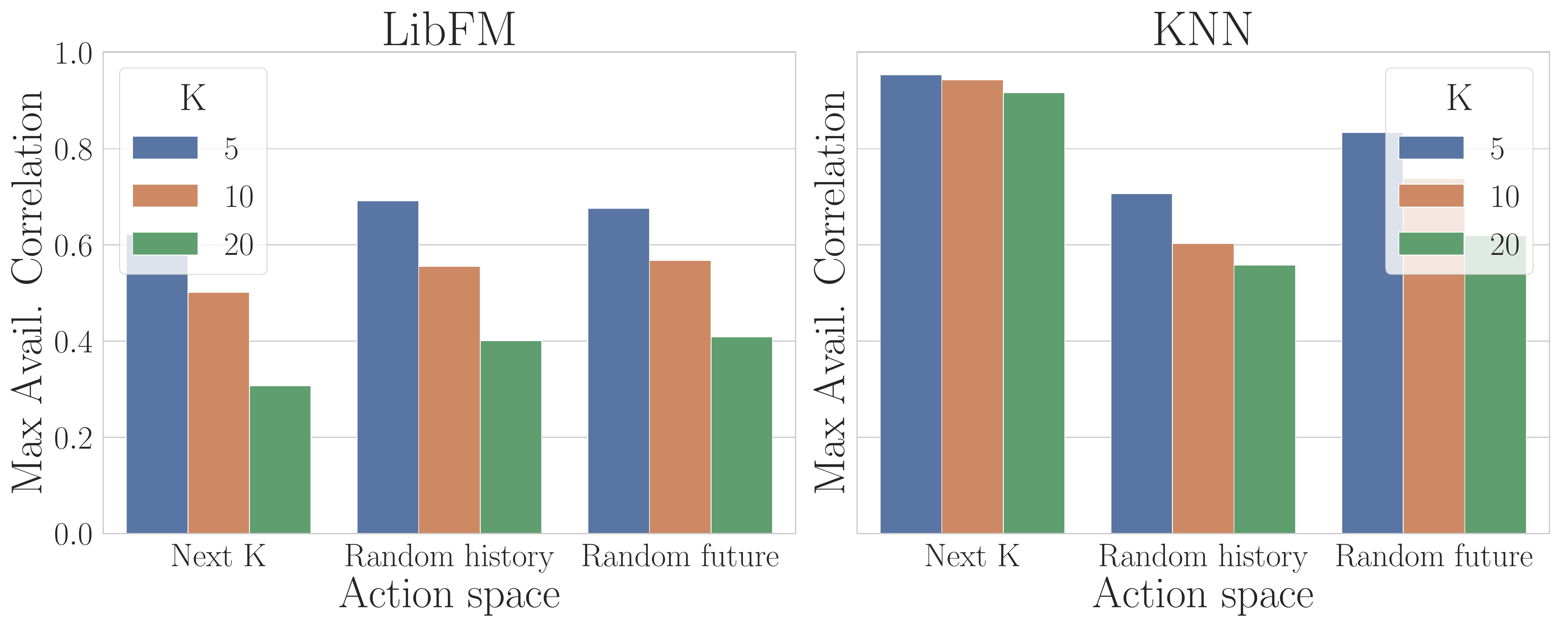}
    \caption{Comparison of Spearman's correlation between item popularity and max availability for different action spaces and models.
    Reachability evaluated on ML-1M with $\beta=2$.} \label{fig:action_pop_bias}
\end{figure}

The impact of user action spaces is displayed in Figure~\ref{fig:action_pop_bias}, which plots the correlation between popularity and max availability for different action spaces.
For comparison, the correlation between popularity and baseline availability is just over 0.8 for all of these settings\footnote{Due to variation in baseline actions, the baseline availability is not exactly the same.}, while the correlation with dataset prevalence is 0.346.
Table~\ref{tbl:corr} shows these correlation values across datasets for a fixed action model. In all cases with the LibFM model, the pattern that popularity is less correlated with max availability than baseline availability holds; however, the correlation with dataset prevalence varies.

To investigate experience bias, we similarly compute the Spearman rank-order correlation coefficient to measure the presence of a monotonic relationship between user experience (as measured by number of items rated) and discovery (either in the baseline or max case).
We observe correlation values of varying sign across datasets and models, and none are particularly strong (Table~\ref{tbl:corr_exp}).

\begin{table}[t]
    \caption{Spearman's correlation with popularity for \emph{Next K} with $K=10$ and $\beta=2$.}
    \label{tbl:corr}
    \vskip 0.15in
    \begin{center}
    \begin{small}
    \begin{sc}
    \begin{tabular}{ll|rrr}
    \toprule
\\ & &  corr. with & corr. with & corr. with \\
dataset & model& dataset prevalence &  baseline availability &  max availability \\
\midrule
 ml-1m &  libfm &       0.346280 &       0.827492 &      0.501316 \\
  ml-1m &    knn &       0.346280 &       0.949581 &      0.942986 \\
   mind &  libfm &       0.863992 &       0.825251 &      0.435212 \\
 lastfm &  libfm &       0.133318 &       0.671101 &      0.145949 \\
\bottomrule
    \end{tabular}
    \end{sc}
    \end{small}
    \end{center}
    \vskip -0.1in
    \end{table}

\begin{table}[t]
    \caption{Spearman's correlation with experience for \emph{Next K} with $K=10$ and $\beta=2$.}
    \label{tbl:corr_exp}
    \vskip 0.15in
    \begin{center}
    \begin{small}
    \begin{sc}
    \begin{tabular}{ll|rrr}
    \toprule
\\ & &  corr. with & corr. with \\
dataset & model& baseline discovery &  max discovery \\
\midrule
 ml-1m &  libfm &       0.475777 &      0.530359 \\
  ml-1m &    knn &       0.206556 &     -0.031929 \\
   mind &  libfm &       0.050961 &      0.112558 \\
 lastfm &  libfm &      -0.084130 &     -0.089226 \\
\bottomrule
    \end{tabular}
    \end{sc}
    \end{small}
    \end{center}
    \vskip -0.1in
    \end{table}

Finally, we investigate gender bias. We compare discovery across user gender for ML-1M and availability across artist gender for LastFM (Figure~\ref{fig:gender_bias}). We do not observe any trends in either baseline or max values.
\begin{figure}
    \center
    \includegraphics[width=0.73\columnwidth]{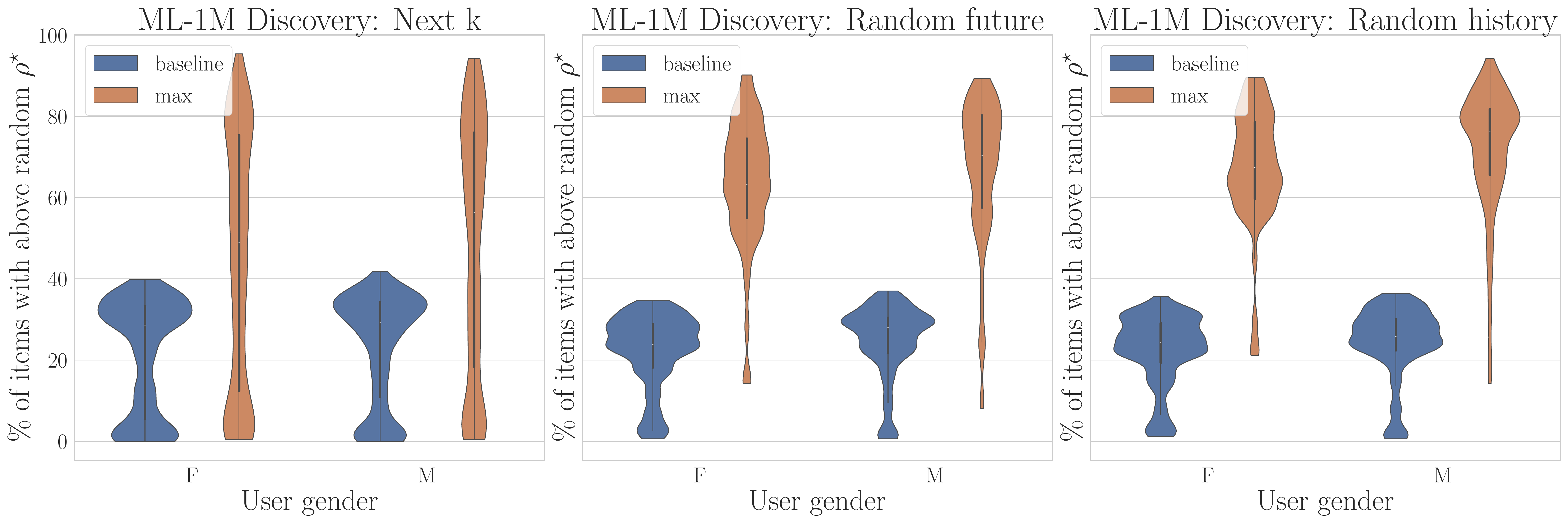}
    \includegraphics[width=0.245\columnwidth]{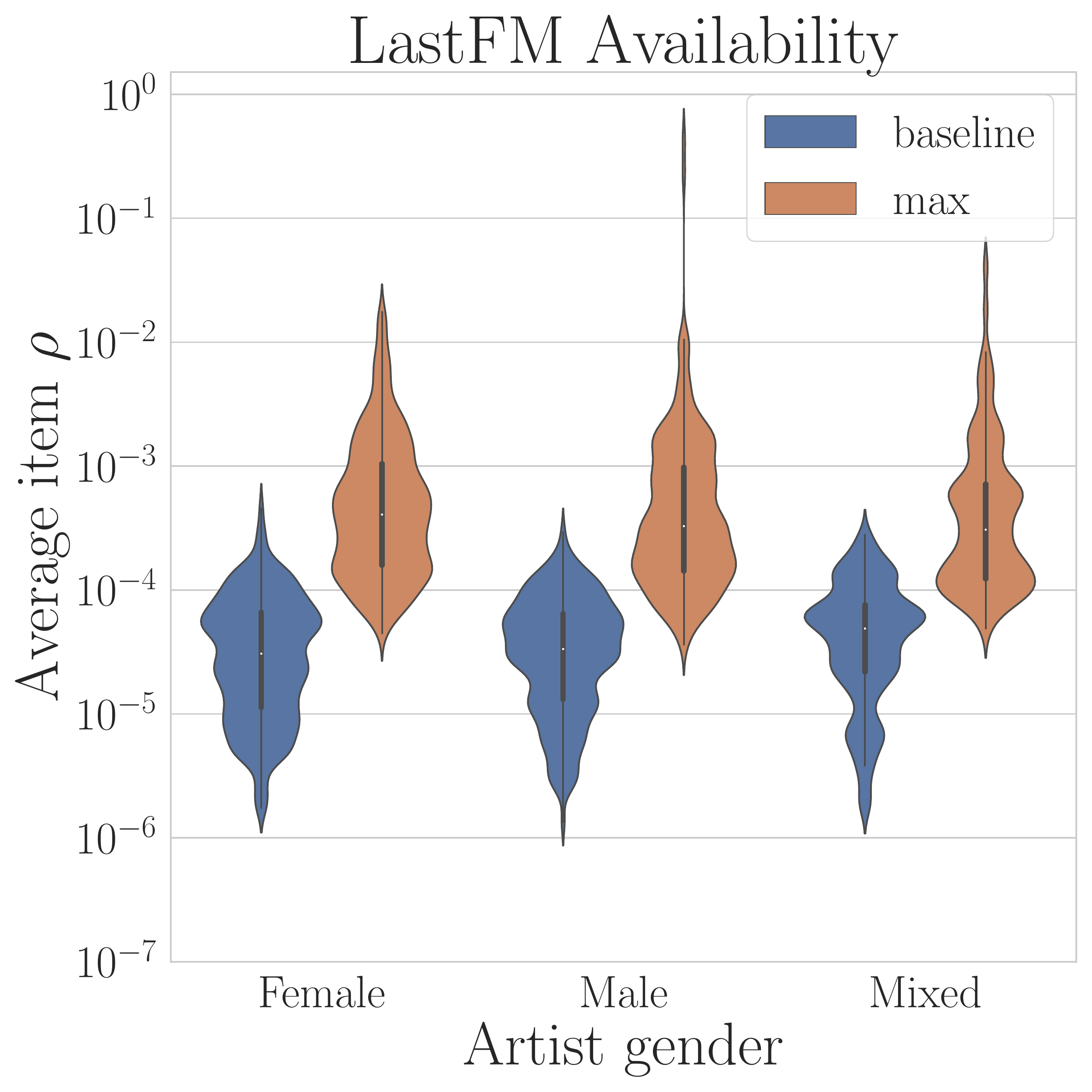}
    \caption{Side by side comparison of baseline and maximum discovery across user gender (left 3 panels) and availability across artist gender (rightmost panel).
    Reachability evaluated on ML-1M and LastFM with LibFM model, $K=10$, different action spaces, and $\beta=2$.} \label{fig:gender_bias}
\end{figure}

\end{document}